\documentclass[12pt]{article}
\usepackage{longtable}
\usepackage[section]{placeins}
\usepackage[top=1in, left=1in, right=1in, bottom=1in]{geometry}
\usepackage[parfill]{parskip}	
\usepackage{float}
\RequirePackage{subfig}
\usepackage[export]{adjustbox}
\usepackage{natbib}
\bibpunct{(}{)}{;}{a}{}{,}
\usepackage[hyphens]{url}
\usepackage[hidelinks]{hyperref}
\hypersetup{breaklinks=true}
\usepackage{lscape}
\usepackage{setspace}
\usepackage{amssymb}
\usepackage{amsmath}
\usepackage{amsfonts}
\usepackage{amsthm}
\newcommand{\set}[1]{\left\{#1\right\}}                     
\newcommand{\abs}[1]{\left|#1\right|}                       
\newcommand{\bra}[1]{\left(#1\right)}

\newtheorem{theorem}{Theorem}[section]
\newtheorem{lemma}[theorem]{Lemma}
\newtheorem{proposition}[theorem]{Proposition}

\begin{document}
\pagestyle{plain}
\def\spacingset#1{\renewcommand{\baselinestretch}%
{#1}\small\normalsize} \spacingset{1}

\newcommand{\blind}{0}

\newcommand{\tit}{\Large Building Representative Matched Samples with Multi-valued Treatments in Large Observational Studies\thanks{For comments and suggestions, we thank Paul Rosenbaum. We acknowledge support from a grant from the Alfred P. Sloan Foundation.}}

\if0\blind
{
  \title{\tit}
  \author{Magdalena Bennett\thanks{Department of Education Policy and Social Analysis, Teachers College at Columbia University; email: \url{mb3863@tc.columbia.edu}.} \and Juan Pablo Vielma\thanks{Operations Research and Statistics Group, Sloan School of Management, Massachusetts Institute of Technology; email: \url{jvielma@mit.edu}.} \and Jos\'{e} R. Zubizarreta\thanks{Department of Health Care Policy, and Department of Statistics, Harvard University; email: \url{zubizarreta@hcp.med.harvard.edu}.}
    }
    
    \date{} 
    
  \maketitle
} \fi

\if1\blind
{
  \bigskip
  \bigskip
  \bigskip
  \begin{center}
    {\Large Building Representative Matched Samples with Multi-valued Treatments in Large Observational Studies}
\end{center}
  \medskip
} \fi

\bigskip
\begin{abstract}

In this paper, we present a new way of matching in observational studies that overcomes three limitations of existing matching approaches.  First, it directly balances covariates with multi-valued treatments without requiring the generalized propensity score.  Second, it builds self-weighted matched samples that are representative of a target population by design.  Third, it can handle large data sets, with hundreds of thousands of observations, in a couple of minutes.  The key insights of this new approach to matching are balancing the treatment groups relative to a target population and positing a linear-sized mixed integer formulation of the matching problem.  We formally show that this formulation is more effective than alternative quadratic-sized formulations, as its reduction in size does not affect its strength from the standpoint of its linear programming relaxation. We also show that this formulation can be used for matching with distributional covariate balance in polynomial time under certain assumptions on the covariates and that it can handle large data sets in practice even when the assumptions are not satisfied.  This algorithmic characterization is key to handle large data sets.  We illustrate this new approach to matching in both a simulation study and an observational study of the impact of an earthquake on educational attainment. After matching, the results can be visualized with simple and transparent graphical displays: while increasing levels of exposure to the earthquake have a negative impact on school attendance, there is no effect on college admission test scores.

\end{abstract}
\noindent%
\vfill

\newpage
\spacingset{1.5} 
\clearpage

\section{Introduction}
\label{sec:introduction}

\subsection{Practical appeal of matching}

In observational studies, matching is a general method for covariate adjustment that approximates the ideal experiment that would be conducted if controlled experimentation was possible.
Under the assumption of strong ignorability \citep{rosenbaum1983central}, or no unmeasured confounders \citep{imbens2004nonparametric}, matching methods are often used to estimate treatment effects in observational studies with binary treatments (e.g., \citealt{dehejia1999causal, haviland2007combining}), longitudinal \citep{lu2005propensity, zubizarreta2014isolation} and multilevel data \citep{li2013propensity, zubizarreta2017optimal}, but also under different assumptions, for example with instrumental variables \citep{baiocchi2010building, zubizarreta2013stronger} or in discontinuity designs \citep{keele2015enhancing, mattei2016regression}.

The practical appeal of matching methods lies in part in the conceptual simplicity and transparency of their adjustments \citep{cochran1973controlling}.
These adjustments are an interpolation instead of an extrapolation based on a model that can be misspecified \citep{rosenbaum1987model}.
Matching also enables the integration of quantitative and qualitative analyses \citep{rosenbaum2001matching}.
Furthermore, since matching methods do not routinely use the outcomes for their adjustments, it is often argued that they promote the objectivity of the study by separating the design and analysis of an observational study into two distinct stages \citep{rubin2008objective}.
Finally, matching methods can facilitate simpler forms of statistical inference and sensitivity analyses to hidden biases (see, for instance, Chapter 3 of \citealt{rosenbaum2010design1}).
See \cite{stuart2010matching}, \cite{imbens2015matching}, and \cite{rosenbaum2017observation} for overviews of matching methods.

\subsection{Three challenges}

Most of the work in matching has been for binary treatments and, although there are methods for more general treatments 
(e.g., \citealt{lu2001matching}, \citeyear{lu2011optimal}; \citealt{yang2016propensity,lopez2017estimation}), 
some challenges remain.
One of these challenges relates to the difficulty of balancing covariates by means of the estimated propensity score \citep{rosenbaum1983central} or generalizations thereof (e.g., \citealt{joffe1999propensity, imbens2000role, imai2004causal}).
The propensity score is the conditional probability of treatment assignment given observed covariates.
It has the important property that matching on the propensity score tends to balance the covariates used to estimate the score; however, in any given data set it might be difficult to balance the covariates even if the propensity score model is correctly specified \citep{yang2012optimal}.
This challenge has been discussed, e.g., by \cite{hill2011bayesian} and \cite{zubizarreta2011matching}, and alternative methods that directly balance the covariates have been proposed (e.g., \citealt{diamond2013genetic, zubizarreta2012using}), yet these are for binary treatments.
The difficulties of balancing covariates by matching on generalizations of the propensity score can easily be exacerbated with multi-valued treatments.

The second challenge with matching methods, both for binary and multi-valued treatments, is targeting parameters of general scientific and policy interest, especially when there is limited overlap in covariate distributions across treatment groups.
When there is sufficient overlap, it is possible to estimate the average effect of one treatment in place of another treatment (possibly the control treatment) on one of these two treatment groups, but it is difficult to estimate other parameters such as the average effect on all treatment groups (in resemblance to the standard average treatment effect; ATE), or on a particular group as defined by covariates (the conditional average treatment effect; CATE), without matching with replacement or reweighting which can make inference more complicated, as discussed by \cite{abadie2008failure}.
When there is limited overlap, it is not possible to target the average effect of one treatment in place of another treatment for any of these treatment groups without imposing strong parametric assumptions, and it is common in practice to settle for more local average effects which may have more internal validity but be less generalizable.
Often times, these analyses are criticized for their limited scientific and policy interest (see, e.g., \citealt{imbens2010better}).

A third challenge relates to matching in large data sets.
Most common optimization-based matching methods rely on quadratic-sized formulations that cannot handle data sets with hundreds of thousands of observations quickly.
As we describe below, the problem is that such formulations are too big to be practical or that they require additional structure on the covariates in order to run quickly.

Through a case study of the impact of an earthquake on educational achievement, our goal in this paper is to overcome these three challenges and present a new matching method that: (i) handles multi-valued treatments and directly balances covariates without estimating a generalization of the propensity score; (ii) finds self-weighting matched samples that are not only balanced but also have a similar structure to the one of a target population (thereby allowing the investigator to target parameters of general policy and scientific interest); and (iii) runs quickly in large data sets, for example, with hundreds of thousands of observations in a couple of minutes.
For this, we use the following ideas.

\subsection{Main ideas}\label{mainideas}

Optimization-based methods for matching with binary treatments can be divided, roughly, into network flow methods (e.g., \citealt{rosenbaum1989optimal,hansen2004full,pimentel2015large}) and integer or mixed integer programming (MIP) methods (e.g., \citealt{nikolaev2013balance,sauppe2014complexity,zubizarreta2012using,zubizarreta2014matching}).
For the most, network flow methods minimize an aggregate measure of covariate distances, and not covariate balance directly, although there is the clever extension by \cite{pimentel2015large} for nested nominal covariates.
The advantage of MIP methods is that they target covariate balance (and distances) more directly and flexibly.
However, MIP based methods require the solution of a theoretically intractable or NP-hard problem, whereas network flow matching methods only require the solution of a polynomially solvable problem that is tractable both in theory and practice.
Still, thanks to state-of-the-art MIP solvers which nearly double their speeds every year \citep{bixby2012brief,Achterberg2013}, MIP methods are mostly tractable in practice, and hence the computational advantage of network flow methods has been steadily decreasing.

If we consider matching with three or more treatments or exposures, this computational advantage vanishes as network flow methods for minimum distance problems also become theoretically intractable \citep{michael1979computers}.
In this case, direct adaptations of both network flow and MIP-based methods are also intractable in practice.
To circumvent this computational limitation, we match each treatment group to a representative sample of a target population and develop a new MIP-based matching formulation.
A key to develop this method is a linear programming (LP) view of the difference in computational complexity between minimum distance matching and MIP-based methods.
This view models the central problem of both methods as a MIP problem and then analyses the strength of their LP relaxation. For minimum distance matching the LP relaxation has the strongest possible \emph{integral} property, which implies the associated MIP problem can be solved as an LP problem in polynomial time \citep{schrijver2003combinatorial}. For MIP based techniques the LP relaxation fails to be integral, but having an LP relaxation that is ``close'' to integral is known to result on small solve times \citep{vielma2015mixed}. Unfortunately, we often face a trade-off between strong, but large formulations and smaller, but weaker formulations.
With this in mind, we show that a linear-sized MIP formulation for covariate balance is as strong as an alternative, and more common, quadratic-sized formulation. We also show that this formulation is in fact integral when only two covariates are considered or when the nested  covariate structure of \cite{pimentel2015large} is present.
While the formulation is not integral when more than three covariates are used, the associated MIP nonetheless remains strong and can be solved in minutes in large data sets.
In contrast, attempting to solve the problem with existing formulations (both network and quadratic-sized MIP formulations) results in an out-of-memory error even in a workstation with 32GB of RAM. 
In this age of big data, this algorithmic characterization is key to handle large data sets in a practical manner.  
 

In addition to allowing a computationally practical approach for matching with multi-valued treatments, the use of this auxiliary reference sample also increases the versatility of the method. By selecting this reference or template sample in different ways, we can target different estimands beyond the common average treatment effect on the treated (ATET), such as the (general) average treatment effect (ATE) or a conditional average treatment effect (CATE) for a particular group of the population (e.g., subjects with a particular age in a particular ethnic group).
In this manner we can build representative matched samples with multi-valued treatments that may be numerous and unordered.
In our procedure, covariate distances between matching play a secondary role (relegated to the rematching of the balanced units, in the spirit of \citealt{zubizarreta2014matching}).
In fact, our procedure does not require explicitly modeling the propensity score and directly balances the covariates by design; this is, as specified before matching by the investigator.
Furthermore, it facilitates checking covariate balance, as one can tabulate and plot the distributions of the observed covariates as opposed to evaluating balance in the context of a model.

\subsection{Case study and outline}

We illustrate all these ideas in the context of an observational study of the impact of a natural disaster on educational achievement; in particular, of the impact of the 2010 Chilean earthquake on test scores of university admission exams.
This is a very important question because due to its location in the Pacific Ring of Fire, Chile is considered one of the most seismically active countries in the world \citep{disaster2017}, concentrating 78 earthquakes above 7.0 magnitude since 1900 \citep{uchileearthquake} and holding the record for the largest earthquake ever registered \citep{usgsearthquakes}.
Furthermore, despite its sustained economic growth and being one of the most robust economies in Latin America, Chile has the highest Gini coefficient among the 35 OECD countries, followed by Mexico and the United States, which makes it the country with the most unequal income distribution of the OECD countries, according to this indicator \citep{oecd2016}.
In this context of extreme income inequality it is important to understand the impact of these recurring natural disasters on standardized university admissions examination scores as they almost fully determine students' entrance to university, one of the main avenues of social mobility and opportunity in Chile \citep{torche2005}.

To address this question, the rest of this paper is organized as follows.
In Section \ref{sec:impact}, we describe the Chilean educational system, the 2010 Chilean earthquake, and a particularly rich data set that is a census of the same students before and after the earthquake.
In Section \ref{sec:representative} we explain how to use the ideas in \cite{silber2014template} for building representative matched samples of target populations with possibly unordered and many treatment or exposure groups.
In Section \ref{sec:effective}, we formally analyze and contrast the properties of linear- and quadratic-sized MIP formulations for distributional balance, and evaluate their performance in data sets of increasing size.
In Section \ref{sec:results} we present estimates of the impact of levels of exposures to the earthquake both on school attendance and standardized test scores for university admission.
In Section \ref{sec:summary} we conclude with a summary and remarks.


\section{Impact of a natural disaster on educational opportunity}
\label{sec:impact}


\subsection{On the Chilean educational system}


In Chile, most students complete the required 12 years of school education \citep{casen2011}.
In their last year of high school, students take a college admission test called \emph{Prueba de Selecci\'on Universitaria} or PSU in order to apply to college.
The PSU is a high-stakes test as it comprises nearly 80\% of the final college admission score and it can only be taken once a year, at the end of the Chilean academic year.
Most students register for the test in their last year of high school, and the vast majority of them takes it only once in their lives.
The importance of the PSU does not only relate to being admitted into a given college, but also to obtaining financial aid and having the possibility to actually attend to college \citep{dinkelman2014}.
In Chile, higher education is one of the main avenues for social mobility and the improvement later life outcomes \citep{torche2005}.
Studies have shown that even though the returns to higher education in Chile are heterogeneous, there are still high, positive returns to attending college, specially highly-selective ones \citep{hastings2013, urzua2015}.

\subsection{The 2010 Chilean earthquake}

In February 27th, 2010, an earthquake of magnitude 8.8 struck near Concepci\'on, Chile's second largest city, going down in history as the  6$^\mathrm{th}$ most severe earthquake registered since 1900 so far \citep{usgs2014largest}.
The disaster severely damaged over 500,000 homes and affected more than 2 million people.
Losses amounted to 18\% of the gross domestic product \citep{mineduc2013recon}.
Nearly 20\% of the schools in the affected regions suffered moderate damages or worse, and many schools had to be completely reconstructed.
In addition, over 40\% of the students in these regions could not start their academic year on time.
24 million US dollars were redirected to rebuild and restore the affected schools and the academic calendar was adjusted \citep{mineduc2013recon}.

\subsection{A longitudinal census of students}

In our analyses, we use a new and rich administrative data set of the same high school students measured before and after the earthquake.
This data set is a longitudinal census as it comprises all Chilean students in 10th grade in 2008 (before the earthquake) and it collects their information again in 12th grade in 2010 (after the earthquake), when they are finishing high school and applying for college.
In 2008, the data provides detailed measures of the students, their schools, and respective households.
For instance, the data provides the standardized test scores from the Education Quality Measurement System (SIMCE), school attendance, GPA ranking within the school, and school characteristics such as its socioeconomic status.
In addition, it offers extensive information about socioeconomic characteristics of the students and their households, such as their parents education and the household income.

In 2010, the data contains two outcomes of interest: (i) the students' school attendance that academic year, and (ii) the language and mathematics PSU scores (both measured after the earthquake).
The first outcome is assesed as percentage.
The scale of the PSU is the same for the two tests and it ranges between 150 and 850 points, with a mean of 500 points and a standard deviation of 110 points \citep{demre2018}.
The data comprises 121279 students measured before and after the earthquake.

\subsection{Earthquake intensity levels}

We use peak ground acceleration (PGA) to measure the strength of the earthquake.
In contrast to other seismic intensity scales, for example, the Mercalli or the Richter scales, PGA is a purely physical measure of the shaking of the earthquake at a given location.
Following \cite{zubizarreta2013effect}, we used the PGA values provided by the United States Geological Survey \citep{usgs2011shakemap} to estimate the PGA in each of the counties where we have data.
Using these values, we created three measures of exposure to the earthquake: one with three levels of shaking, another with five levels, and a final one with ten, basically defined by quantiles of exposure (see Appendix A in the Supplementary Materials for details).


\section{Representative matching with multi-valued treatments}
\label{sec:representative}

As noted in Section~\ref{mainideas}, minimum distance matching (arguably, the mainstream form of optimal matching) with two treatments is computationally tractable (i.e. polynomially solvable), but it is computationally intractable (NP-hard) for three or more treatments.
To circumvent this limitation, instead of explicitly matching samples across treatments, we will separately match each treatment sample to a representative or template sample of a target population of special interest, extending the ideas of \cite{silber2014template}.
This will allow us to address the two aforementioned difficulties of handling multi-valued treatments and building representative samples in matching.
In the following section, we incorporate the idea into a MIP-based matching approach for large data sets.
The basic idea is depicted in Figure \ref{figure1}.

\begin{figure}
\centering
\caption{Template matching for multi-valued treatments.}
\subfloat[Traditional matching]{
\centering
\label{figure1a}
\makebox{
\begin{adjustbox}{minipage=\linewidth,scale=0.58}
\includegraphics[width=0.9\textwidth]{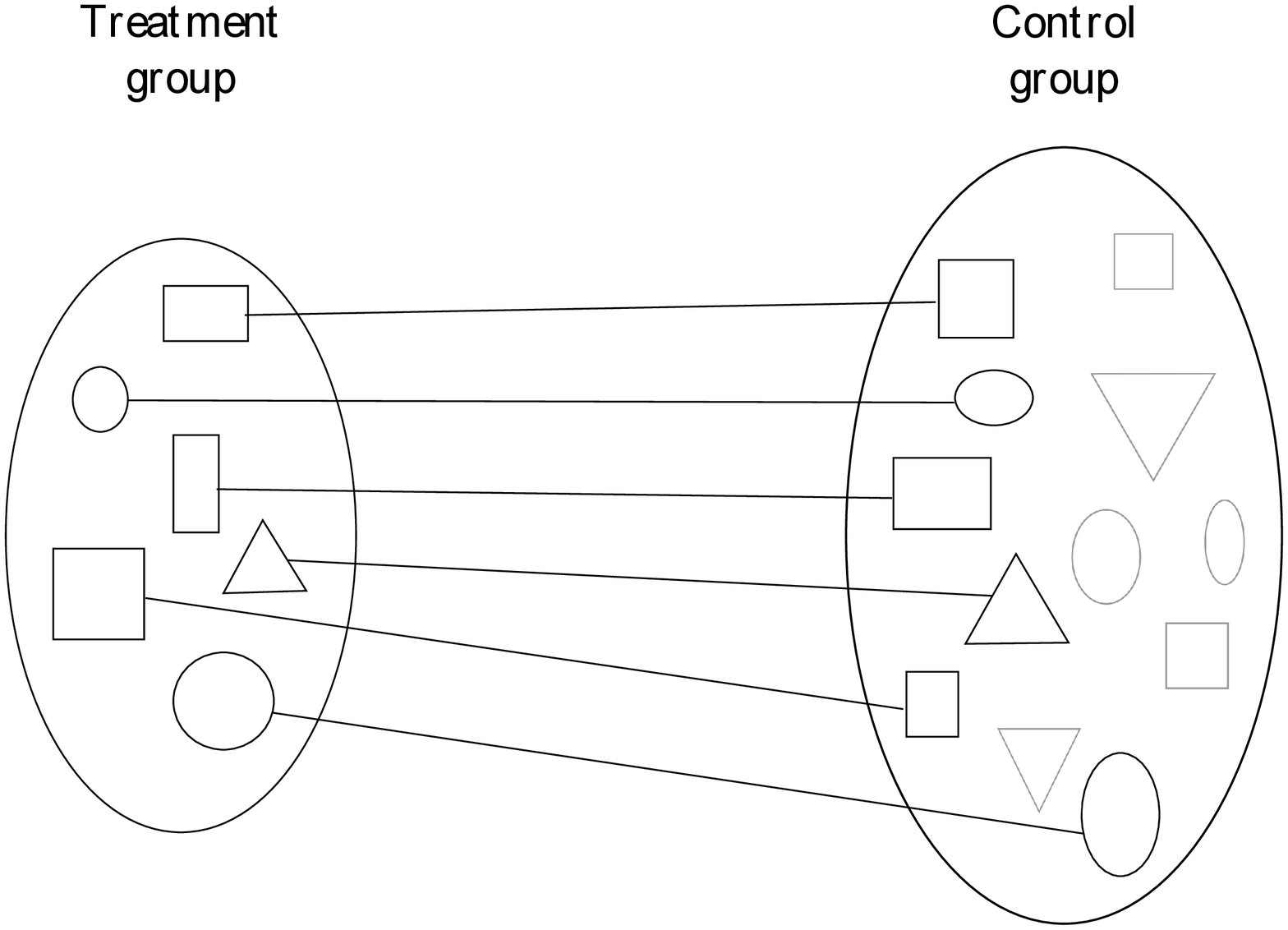}
\end{adjustbox}
}
}\\
\subfloat[Template matching for treatment and control groups]{
\centering
\vspace{.2cm}
\makebox{
\begin{adjustbox}{minipage=\linewidth,scale=0.58}
\includegraphics[width=0.9\textwidth]{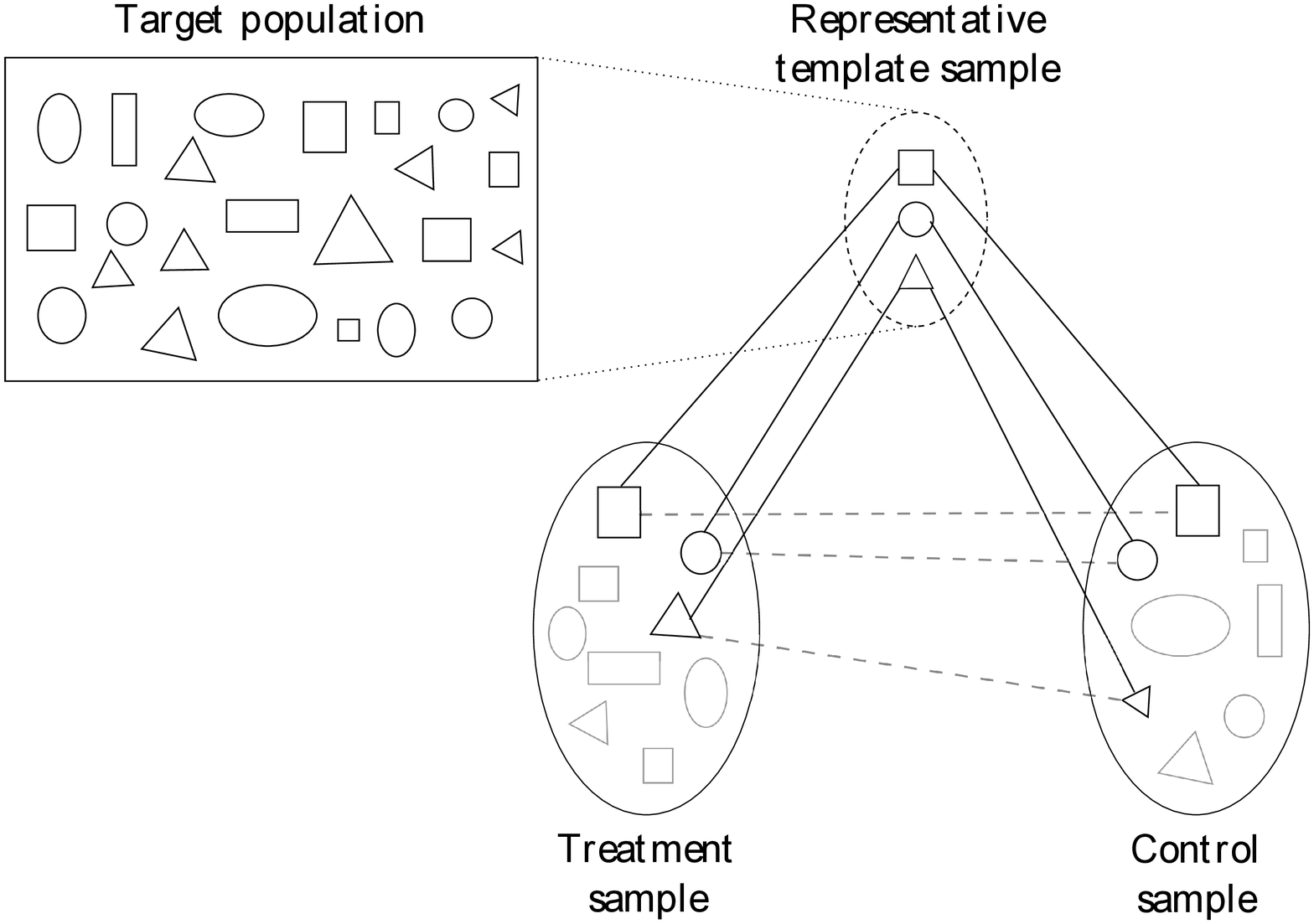}
\end{adjustbox}
}
\label{figure1b}
}\\
\subfloat[Template matching for multi-valued (ordered or unordered) treatments, e.g. levels of exposure]{
\centering
\vspace{.4cm}
\makebox{
\begin{adjustbox}{minipage=\linewidth,scale=0.58}
\includegraphics[width=0.9\textwidth]{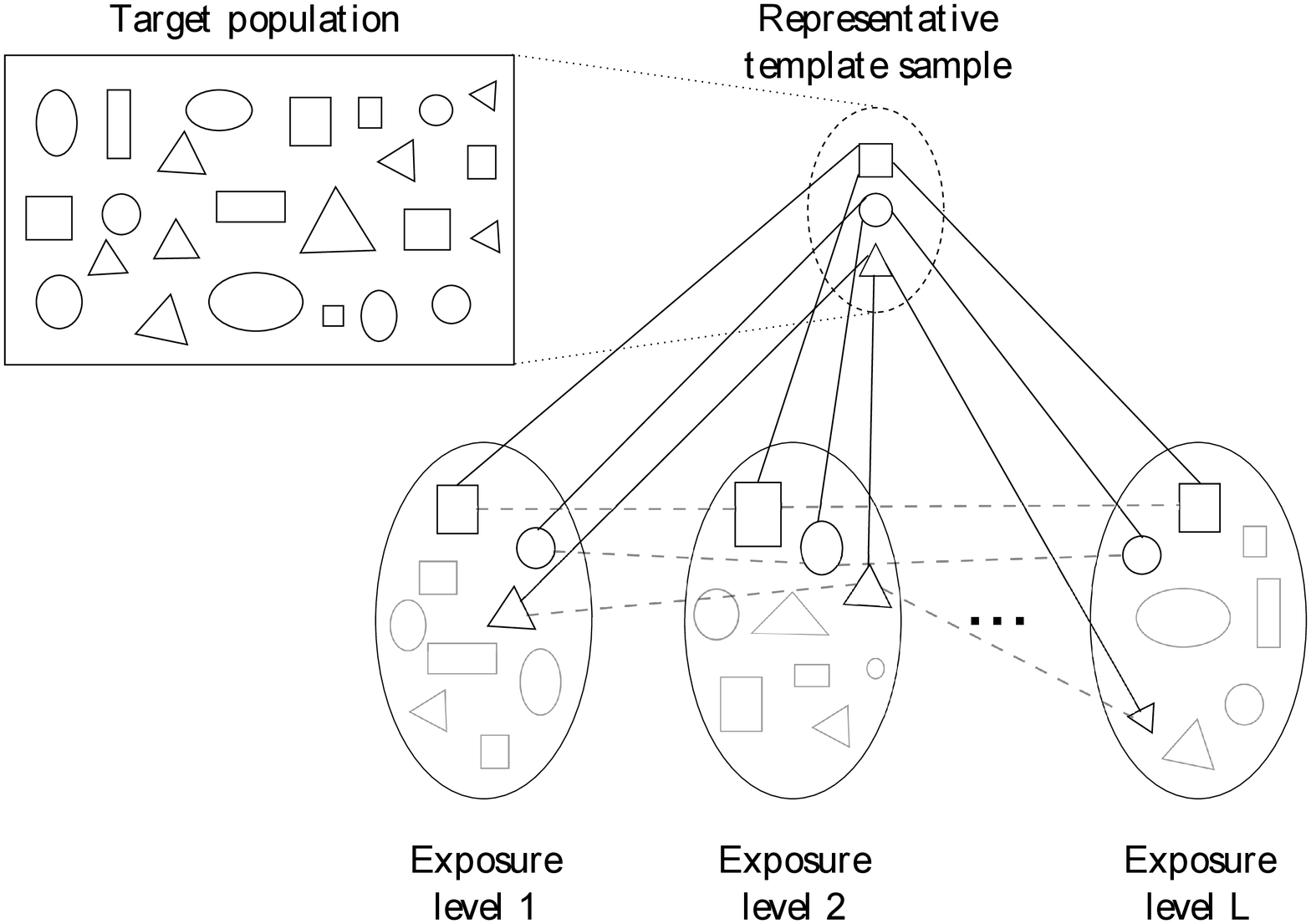}
\end{adjustbox}
}
\label{figure1c}
}
\label{figure1}
\end{figure}

Figure \ref{figure1}(a) shows traditional bipartite matching with a smaller treatment group than a control group.
The goal is to match all the treatment and control units with similar observed covariates, represented in the figure by the shapes of the elements.
Here, the target of inference is the average treatment effect on the treated (ATET).
But how to target a different estimand, one for a particular population of policy or scientific interest?
Call this estimand the target average treatment effect (TATE; \citealt{kern2016assessing}).
Figure \ref{figure1}(b) illustrates how to extend template matching for this purpose.
Here, a representative template sample is selected from the target population of interest and then matched to the treatment and control groups separately.
By construction, the matched groups will be balanced relative to the template sample and therefore to each other.
We can repeat this process with multi-valued (ordered or un-ordered) treatments, possibly tens or hundreds treatments --- as many as the data allows us to balance.
This is illustrated in Figure \ref{figure1}(c).

We select the template sample by drawing random samples of a given size from the population of interest.
In our case study, the population of interest was the population of all high school students in tenth grade in 2008 in Chile, and we drew 500 random samples of size 1000 from this population.
We chose this sample size so that it was large enough to represent all levels of exposure to the earthquake.
In this paper, we fixed this sample size for simplicity of exposition, but the matching formulations in Section \ref{sec:effective} can be extended to relax this requirement and its implied fixed matching ratio by finding the largest sample for each treatment that is balanced relative to certain moments of the empirical distribution of the covariates that characterizes the target population.
From the 500 random samples, we selected the sample that was closest to the population in terms of a robust version of the Mahalanobis distance (\citealt{rosenbaum2010design1}, Chapter 8).
Table \ref{template_comparison} in the Supplementary Materials describes the population and the selected template sample in terms of the means for each category of the observed covariates.
In our case study, we used this template to match representative samples for each level of exposure to the earthquake as follows.

As described in Table \ref{template_comparison}, we have 14 categorical covariates with 78 categories in total, or 78 binary covariates for each of these categories.
These covariates are: the student's gender, and ethnicity (3 categories); the father and mother's education (5 categories each), household income (7 categories), and number of books at home (6 categories); student's school attendance (10 categories), GPA ranking within the school (10 categories), and SIMCE score (10 categories); the school's type (public, voucher, or private), location (rural or urban), whether the school is catholic or not, its socioeconomic status (SES; 5 categories), and the average SIMCE score of the school (10 categories).
We perfectly balanced the marginal distributions of all these covariates by matching with fine balance \citep{rosenbaum2007minimum}.
Fine balance perfectly balances the marginal distributions of the observed covariates, but without constraining units to be matched within the same categories of the covariates (see \citealt{visconti2018handling} for an illustration).
By matching with fine balance to the template, we can guarantee that the matching across multiple values of the exposure will be perfectly balanced relative to each other and to template of the population in terms of the observed covariates (see Table \ref{tab:balance} below).

In our case study, identification of the TATE relies on strong ignorability \citep{rosenbaum1983central, imbens2015causal} of both the treatment assignment and the sample selection mechanisms (see, e.g., \citealt{tipton2013improving} and \citealt{kern2016assessing}).
The fact that the entire territory of the country is prone to earthquakes that people cannot anticipate, the broad scope of the covariates in our data set for which we adjust by matching, and the act of randomly selecting the template sample from the (entire) population, make these assumptions plausible.
We also assume the stable unit treatment value assumption (SUTVA; \citealt{rubin1980randomization,rubin1986comment}) holds.


\section{Effective formulations for matching using integer programming}
\label{sec:effective}

In this section, we analyze and contrast the effectiveness linear- and quadratic-sized mixed integer programming matching formulations.
In particular, we show that a particular linear-sized formulation is more effective as it permits us to handle considerably larger data sets in a shorter time.
For example, in our case study, one of the data sets involved more than 70,000 units, yet the linear-sized formulation could find the solution in less than one minute in a standard laptop computer.
With quadratic-sized formulations, the problem did not even fit in memory.
Coupled with the matching strategy in the above section, this allows us to build balanced and representative matched samples with multi-valued treatments in considerably larger data sets.

Parts of this section are technical and so we begin by describing its structure and main results.
In Section \ref{sec:new} we describe the notation and a quadratic-sized formulation for distributional balance that works well in some instances, but that cannot handle the data in our case study due to its size.
In Section \ref{sec:reducing} we describe how the quadratic-sized formulation can be reduced to a linear-sized formulation through a simple procedure that could potentially result in a loss of formulation strength.
In Section \ref{sec:complexity} we show that no strength is actually lost with this reduction and prove three results: the linear programming (LP) relaxation (i.e., the ``strength'') of the linear-sized formulation is equal to the one of the quadratic-sized formulation; with two covariates or nested covariates, the LP relaxation of both formulations is integral; and with three or more covariates, the LP relaxation of both formulations can fail to be integral.
In practice, however, the smaller formulation is quite practical.
As we show in Section \ref{sec:performance}, this formulation can handle data sets with more than 700,000 units in only a couple of minutes.

\subsection{A large formulation for distributional balance}
\label{sec:new}

Let $\mathcal{T}=\set{t_1,\ldots,t_T}$ represent the $T$ units in the template sample, $\mathcal{L}=\set{\ell_1,\ldots,\ell_{L}}$ be the $L$ units in the sample under a given treatment or exposure level, and
$\mathbb{L}$ be the family containing such sets $\mathcal{L}$ of units for each treatment or exposure level (which we do not explicitly index to emphasize that exposure levels do not need to be ordered).
In the following, we fix $\mathcal{L} \in \mathbb{L}$ because, as explained in the previous section (Figure \ref{figure1}(c)), the matching process is identical for every exposure level $\mathcal{L} \in \mathbb{L}$.
In addition, let  $\mathcal{P}=\set{p_1,\ldots,p_P}$ denote the $P$ observed covariates that we aim to balance and $\mathbf{x}_i=\bra{x_{i,p}}_{p\in \mathcal{P}}$ define the covariate values for the units in the template and exposure level samples for $i \in \mathcal{T}$ and $i \in \mathcal{L}$, respectively.
Finally,  let $\mathcal{K}(p)=\set{k_1,\ldots,k_{K_p}}$ stand for the categories of covariate $p\in \mathcal{P}$, specify the template (level $l$) units in category $k$ in covariate $p$ as $ \mathcal{T}_{p,k}=\{t\in \mathcal{T}\,:\, x_{t, p}=k\}$ ($ \mathcal{L}_{p,k}=\{\ell\in \mathcal{L}\,:\, x_{\ell, p}=k\}$) and let $N_{p,k}=\abs{\mathcal{T}_{p,k}}$.

We can minimize the imbalances in the marginal distributions of the $P$ covariates, through the following mixed integer programming problem
\newlength{\mylengthb}
  \settowidth{\mylengthb}{$\displaystyle\abs{\sum\nolimits_{\ell \in \mathcal{L}_{p,k}} \sum\nolimits_{t\in \mathcal{T}}    m_{t,\ell} - N_{p,k}}$}
  \begin{subequations}\label{original}
\begin{alignat}{5}
               &\underset{\boldsymbol{v},\;\boldsymbol{m}}{\operatorname{minimize}}\quad &\makebox[\mylengthb][l]{$\displaystyle\sum\nolimits_{p\in \mathcal{P}} \sum\nolimits_{k\in \mathcal{K}(p)}v_{p,k}$}&&&\\
      &\text{subject to}\quad & \abs{\sum\nolimits_{\ell \in \mathcal{L}_{p,k}} \sum\nolimits_{t\in \mathcal{T}}  m_{t,\ell} - N_{p,k}}&\leq v_{p,k},  &\quad&\forall p\in \mathcal{P},\quad k\in\mathcal{K}(p),	\label{orilpaa}\\
  &&       \sum_{t\in \mathcal{T}}  m_{t,\ell} &\leq 1,  &\quad& \forall \ell\in\mathcal{L}, \label{orilpstart} \\
    &&      \sum_{\ell\in \mathcal{L}}  m_{t,\ell} &= 1,  &\quad& \forall t \in \mathcal{T},\label{orilpstart2}\\
        &&    m_{t,\ell} &\in \{0 , 1\}, &\quad& \forall t\in \mathcal{T},\quad \ell \in \mathcal{L},\label{oriint}
 \end{alignat}
\end{subequations}
where $v_{p,k}$ defines the imbalance or violations from fine balance \citep{rosenbaum2007minimum} for category $k$ of covariate $p$, and $m_{t,\ell}$ is the binary decision variable which takes the value of 1 if template unit $t$ is matched to level $\mathcal{L}$ unit $\ell$, and $0$ otherwise.

Formulation \eqref{original} has $T\times L + \sum_{p\in \mathcal{P}} K_p$ variables and $T + L+ \sum_{p\in \mathcal{P}} K_p$ constraints excluding the variable bounds.
While this problem size is polynomial on the number of units, the quadratic term $T\times L$ can be prohibitive in practice.
For instance, in our case study one of the data sets results in $1.8\times 10^6$ decision variables.
This makes the problem very hard in practice and attempting to solve it with a state-of-the-art solver, Gurobi v7.5.1 on a Core i7-3770 3.40GHz workstation with 32GB of RAM, results in an out-of-memory error.
We get the same error even if we attempt to solve the much easier LP relaxation obtained by relaxing constraint \eqref{oriint} into $m_{t,\ell} \in [0 , 1]$.

\subsection{Reducing the size of the formulation}
\label{sec:reducing}

One approach to reduce the size of formulation \eqref{original} is to define variables $z_\ell=\sum\nolimits_{t \in \mathcal{T}}m_{t,\ell}$ for each $\ell \in \mathcal{L}$ and use these new variables to eliminate the $m_{t,\ell}$ variables from formulation (\ref{original}).
For this, we first sum constraints \eqref{orilpstart2} to obtain
\begin{subequations}\label{originalsimplified}
\begin{alignat}{5}
             &\underset{\boldsymbol{v},\;\boldsymbol{m}}{\operatorname{minimize}}\quad &\makebox[\mylengthb][l]{$\displaystyle\sum\nolimits_{p\in \mathcal{P}} \sum\nolimits_{k\in \mathcal{K}(p)}v_{p,k}$}&&&\\
    &\text{subject to}\quad & \abs{\sum\nolimits_{\ell \in \mathcal{L}_{p,k}} \sum\nolimits_{t\in \mathcal{T}}  m_{t,\ell} - N_{p,k}}&\leq v_{p,k},  &\quad&\forall p\in \mathcal{P},\quad k\in\mathcal{K}(p)\\
&&       \sum_{t\in \mathcal{T}}  m_{t,\ell} &\leq 1,  &\quad& \forall \ell\in\mathcal{L} \label{leqoneoriginalsimplified}\\
  &&    \sum_{t\in \mathcal{T}} \sum_{\ell\in \mathcal{L}}    m_{t,\ell} &= T,  &\quad& \label{sumconstoriginalsimplified}\\
      &&    m_{t,\ell} &\in \{ 0, 1 \}, &\quad& \forall t\in \mathcal{T},\quad \ell \in \mathcal{L}.
\end{alignat}
\end{subequations}
We claim that formulation \eqref{original} is equivalent to \eqref{originalsimplified}.
For that, first note that any solution to \eqref{original} is a solution to \eqref{originalsimplified}.
Second, suppose that a solution $( \boldsymbol{v}^*,\;\boldsymbol{m}^*)$ to \eqref{originalsimplified} has some $t'\in \mathcal{T}$ such that $\sum_{\ell\in \mathcal{L}}  m^*_{t',\ell} \geq 2$.
Then, because of \eqref{sumconstoriginalsimplified}, there is some $t''\in \mathcal{T}$ with $t''\neq t'$ such that $\sum_{\ell\in \mathcal{L}}  m^*_{t'',\ell} = 0$.
If $\ell' \in \mathcal{L}$ is such that $m^*_{t',\ell'}=1$ we can change the solution by letting $m^*_{t',\ell'}=0$ and $m^*_{t'',\ell'}=1$. This does not change $\sum\nolimits_{t \in \mathcal{T}}m^*_{t,\ell}$ for any $t\in \mathcal{T}$ so the solution remains feasible for \eqref{originalsimplified}. If we repeat this step, we eventually get a solution to \eqref{original}.

Noting that the order of the sums in \eqref{sumconstoriginalsimplified} can be exchanged,   we can replace every occurrence of $\sum\nolimits_{t \in \mathcal{T}}m_{t,\ell}$ by $z_\ell$ in \eqref{originalsimplified} to obtain the following equivalent formulation
 \begin{subequations}\label{projected}
 \begin{alignat}{3}
  &\underset{\boldsymbol{v},\;\boldsymbol{z}}{\operatorname{minimize}} & \quad  \sum\nolimits_{p\in \mathcal{P}} \sum\nolimits_{k\in \mathcal{K}(p)}v_{p,k}  \\
       &\text{subject to}\quad  & \abs{ \sum\nolimits_{\ell \in \mathcal{L}_{p,k}}   z_\ell - N_{p,k}}   &\leq v_{p,k}  &\quad&\forall p\in \mathcal{P},\quad k\in\mathcal{K}(p)\label{prjlpaa}\\
     &      & \sum\nolimits_{\ell \in \mathcal{L}}   z_\ell    &=T  &\quad&\forall p\in \mathcal{P},\quad k\in\mathcal{K}(p)\label{sumtoTconstraint}\\
       &      &  z_\ell &\in \set{0,1} &\quad& \ell \in \mathcal{L},\label{prjlpend}
  \end{alignat}
 \end{subequations}
where $z_\ell$ is a binary decision variable that takes the value 1 if unit $\ell$ of the exposure level group $\mathcal{L}$ is matched to the template $\mathcal{T}$, and $0$ otherwise.
Formulation \eqref{projected} has $ L + \sum_{p\in \mathcal{P}} K_p$ variables and $1+ \sum_{p\in \mathcal{P}} K_p$ constraints excluding the variable bounds, which is significantly smaller than formulation \eqref{original}.
Indeed, for the aforementioned data set in our case study, going from formulation \eqref{original} to \eqref{projected} yields a reduction in the number of variables from over $1.8\times 10^6$ to $19208$ and a reduction in the number of constraints from $19370$ to $162$.
Thanks to this drastic decrease in size we can solve the LP relaxation of \eqref{projected} for this data set in less than one second.

However, this reduction in problem size and solve time for the LP relaxation may not necessarily translate into a reduction in solve time for the MIP formulation \eqref{projected}.
To analyze this solve time, it is convenient to consider the tractability of minimum distance matching (e.g., \citealt{rosenbaum1989optimal}) from an LP perspective instead of the traditional graph theoretical one (i.e., the assignment problem or bipartite matching one).
From this perspective, the tractability of minimum distance matching can be attributed to the fact that the LP relaxation of its standard assignment problem formulation has extreme points or basic feasible solutions that always satisfy the integrality constraints of the formulation.
A formulation with this property is usually denoted \emph{integral} and this property implies that solving the formulation is equivalent to solving its LP relaxation and hence minimum distance matching is polynomially solvable \citep{schrijver2003combinatorial}.
In contrast, the extreme points of the  LP relaxations of \eqref{original} and \eqref{projected}  do not necessarily satisfy the integrality constraints.
Hence, these problems are not necessarily polynomially solvable and may need to be solved with a MIP solver instead of just an LP solver.
In such cases, a relatively good predictor of the difficulty of solving the optimization problem with a MIP solver is the distance between the optimal value of the original MIP problem and its LP relaxation, which is often denoted the \emph{integrality} or \emph{LP GAP} \citep{vielma2015mixed}.
If the GAP is small, the formulation is \emph{strong} and it is expected to lead to faster solve times than similar sized \emph{weaker} formulations with a larger GAPs.
Unfortunately, the elimination of variables through  $z_\ell=\sum\nolimits_{t \in \mathcal{T}}m_{t,\ell}$ to obtain \eqref{projected} from \eqref{original} does not automatically guaranteed to preserve formulation strength, and hence the improvements from size reduction could be negated by a loss of strength.

\subsection{Complexity and strength of the formulation}
\label{sec:complexity}

Using results from \cite{balas1983perfectly} we can fortunately show that no formulation strength is lost in the transformation from \eqref{original} to \eqref{projected}. We formalize this fact in the following simple proposition for which we give a self-contained proof.
\\

\begin{proposition}\label{projectionprop}
  The LP relaxations of formulations \eqref{original} and \eqref{projected} are equivalent.
\end{proposition}
\begin{proof}
We follow a similar logic to our arguments for the equivalence between \eqref{original} and \eqref{originalsimplified}, but using continuous instead of binary variables.

First, let $\bra{\boldsymbol{v},\;\boldsymbol{m}}$ be a feasible solution to the LP relaxation of \eqref{original}. Then $\bra{\boldsymbol{v},\;\boldsymbol{z}}$ for $\boldsymbol{z}$ given by $z_\ell=\sum_{t\in \mathcal{T}}  m_{t,\ell}$ for all  $\ell \in \mathcal{L}$ is a feasible solution to the LP relaxation of \eqref{projected}  obtained by relaxing \eqref{prjlpend} into $ z_\ell \in [0,1]$.

For the converse, let $\bra{\boldsymbol{v},\;\boldsymbol{z}}$ be a feasible solution to the LP relaxation of \eqref{projected}. We obtain a feasible solution $\bra{\boldsymbol{v},\;\boldsymbol{m}}$  to the LP relaxation of \eqref{original} by constructing $\boldsymbol{m}$ as follows. Let $\underline{s}_1=0$,  $\overline{s}_1=\min\set{s\geq 1\,:\, \sum_{i=\underline{s}_1+1}^{s} z_{\ell_i} \geq 1}$, $m_{t_1,\ell_{\overline{s}_1}}=1-\sum_{i=\underline{s}_1+1}^{\overline{s}_1-1} z_{\ell_i}$ and for each $i\in \set{\underline{s}_1+1,\ldots,\overline{s}_1-1}$ let $m_{t_1,\ell_i}=z_{\ell_i}$. Then, for each $j\in\set{2,\ldots,T}$ let $\underline{s}_j=\overline{s}_{j-1}$, $\overline{s}_j=\min\set{s\geq s_{j-1}\,:\, \bra{z_{\ell_{\underline{s}_{j}}}-m_{t_{j-1},\ell_{\underline{s}_{j}}}}+\sum_{i=\underline{s}_j+1}^{s} z_{\ell_{i}} \geq 1}$,
 $m_{t_j,\ell_{\overline{s}_j}}=1-\sum_{i=\underline{s}_j+1}^{\overline{s}_j-1} z_{\ell_i}- \bra{z_{\ell_{\underline{s}_{j}}}-m_{t_{j-1},\ell_{\underline{s}_{j}}}}$, $m_{t_j,\ell_{\overline{s}_j}}=\bra{z_{\ell_{\underline{s}_{j}}}-m_{t_{j-1},\ell_{\underline{s}_{j}}}}$, and for each $i\in \set{\underline{s}_j+1,\ldots,\overline{s}_j-1}$ let $m_{t_j,\ell_i}=z_{\ell_i}$. Because of \eqref{sumtoTconstraint} we have that $m_{t_T,\ell_{\overline{s}_T}}=z_{\ell_{\overline{s}_T}}$ and $z_{\ell_{i}}=0$ for all $i>\overline{s}_T$. Then, by construction $m$ satisfies \eqref{orilpstart}--\eqref{orilpstart2}, $m_{t,\ell} \in [0 , 1]$ for all $t\in \mathcal{T}$ and  $\ell \in \mathcal{L}$,  and $z_\ell=\sum_{t\in \mathcal{T}}  m_{t,\ell}$ for all  $\ell \in \mathcal{L}$. Finally, because of this last equation and the fact that $(z,v)$ satisfies \eqref{prjlpaa} we have that $\bra{\boldsymbol{v},\;\boldsymbol{z}}$ satisfies \eqref{orilpaa}.
\end{proof}
Proposition~\ref{projectionprop} shows that we do not loose formulation strength with the size reduction, but it does not tell us how strong the equivalent formulations are. The following proposition, that we prove in Appendix C in the Supplementary Materials, shows that for problems with at most two covariates or with nested covariates the LP relaxations of both formulations are integral and that, similar to the assignment formulations for propensity score matching, their solution is equivalent to the solution of their LP relaxations.

\begin{proposition}\label{allcovariatesprop}
The  LP relaxations of \eqref{original} and \eqref{projected} are integral (i.e. they have integral extreme points) if
\begin{enumerate}
  \item\label{allcovariatespropcase1} $P\leq 2$, or
  \item\label{allcovariatespropcase2} for  all $i<j$ and $k\in \mathcal{K}\bra{p_j}$ there exist $k'\in \mathcal{K}\bra{p_i}$ such that $\mathcal{L}_{p_j,k}\subseteq \mathcal{L}_{p_i,k'}$.
\end{enumerate}
In particular, under these conditions \eqref{original} and \eqref{projected} can be solved in polynomial time by solving their LP relaxations.
\end{proposition}

Unfortunately, the following lemma, that we prove in Appendix B, shows that the formulations can fail to be integral for non-nested covariates even if $P=3$.

\begin{lemma}\label{examplelemma} Even for $P=3$ and $K_1=K_2=K_3=3$, there exist covariates for which the LP relaxations of \eqref{original} and \eqref{projected}  fail to be integral.
\end{lemma}

This loss of integrality for more general covariate structures is not surprising as the problem is NP-hard. Fortunately,  enough of the formulation strength is preserved so that formulation \eqref{projected} remains extremely effective.
Indeed the size and strength of formulation \eqref{projected} allows to solve all the matching problems in our case study in a few seconds in a regular laptop (see Table \ref{match_time}) even though they do not satisfy the covariate assumptions of Proposition~\ref{allcovariatesprop}.


\subsection{Performance of the linear-sized formulation}
\label{sec:performance}

To evaluate the performance of the linear-sized formulation, we implemented it in the new function \texttt{cardmatch} in \texttt{designmatch} for \texttt{R} \citep{zubizarreta2018designmatch} and tested it in data sets of different sizes.
We increased the largest exposure sample in our data set ($L = 70118$) up to ten times ($L = 701180$) by creating random copies of it.
For each copy, we randomly modified all the nominal covariates by adding 1, 0, or -1, and truncating the resulting values to preserve the original ranges.
We also created templates of different sizes of up to 10000 observations.

Table \ref{per} shows the computing times.
The largest data set we considered had a template sample size of $T = 10000$ and an exposure sample size of $L = 701180$, and took approximately three minutes.
Most of the matchings took less than two minutes.
Naturally, the computing time tends to increase both with the template and exposure sample sizes.

Table \ref{match_time} shows the matching time in our case study with 3, 5, and 10 levels of exposure, and a template sample size of $T = 1000$.
The table shows that all matching times are well under a minute, and most of them are under 10 seconds.
The computing times do not increase monotonically with the sample size in part due to random variation in the generation of the covariate values.
We also compared the performance of the linear-sized formulation in \texttt{cardmatch} to other matching packages based on quadratic-formulations, such as \texttt{optmatch} \citep{hansen2007flexible} and \texttt{rcbalance} \citep{pimentel2015large}.
Trying to run \texttt{optmatch} results in an out-of-memory error, even for our original case study sample ($T = 1000$ and $L = 70118$).
In the case of \texttt{rcbalance}, we were only able to match the original sample ($L = 70118$) to template sizes up to size $T = 6000$. 
For larger template sizes, we were not able to get a solution in the allowed time of 8 hours.  

\begin{table}[H]
\begin{center}
\caption{Computing time (in minutes) of the proposed matching formulation as implemented in the new function \texttt{cardmatch} in \texttt{designmatch} for \texttt{R} \citep{zubizarreta2018designmatch}. We considered different template samples sizes (ranging from 1000 to 10000) and varying exposure sample sizes, increasing from the largest exposure sample size in our case study, 70118, to $10 \times 70118 = 701180$.\label{per}}
\begin{adjustbox}{scale=0.7}
\begin{tabular}{ccccccccccc}
  & \multicolumn{10}{c}{Exposure size $L$}\\
Template size $T$ & 70118 & 140236 & 210354 & 280472 & 350590 & 420708 & 490826 & 560944 & 631062 & 701180 \\ 
  \hline 
1000 & 0.28 & 0.50 & 0.65 & 0.79 & 1.11 & 1.20 & 1.49 & 2.13 & 2.58 & 2.63 \\ 
  2000 & 0.20 & 0.72 & 0.91 & 1.14 & 1.49 & 1.56 & 1.87 & 2.20 & 2.53 & 2.67 \\ 
  3000 & 0.19 & 0.73 & 1.08 & 1.37 & 1.62 & 1.51 & 2.02 & 2.26 & 2.53 & 3.15 \\ 
  4000 & 0.22 & 0.44 & 1.09 & 1.57 & 1.74 & 1.98 & 2.00 & 2.29 & 2.48 & 2.62 \\ 
  5000 & 0.18 & 0.33 & 0.87 & 1.26 & 1.52 & 1.94 & 3.05 & 1.73 & 2.93 & 3.51 \\ 
  6000 & 0.26 & 0.47 & 0.64 & 1.66 & 2.07 & 2.40 & 2.78 & 2.94 & 3.18 & 3.04 \\ 
  7000 & 0.18 & 0.36 & 0.56 & 0.76 & 1.62 & 2.09 & 2.28 & 2.36 & 2.71 & 8.54 \\  
8000 & 0.25 & 0.40 & 0.57 & 0.82 & 1.87 & 2.25 & 2.42 & 2.95 & 3.08 & 3.66 \\ 
  9000 & 0.25 & 0.46 & 0.74 & 0.82 & 0.99 & 2.18 & 2.94 & 3.13 & 4.13 & 3.85 \\ 
  10000 & 0.19 & 0.39 & 0.63 & 0.83 & 1.08 & 2.55 & 2.58 & 2.93 & 3.13 & 3.42 \\  
   \hline
\end{tabular}

\end{adjustbox}
\end{center}
\end{table}

%
%

\begin{table}[H]
\begin{center}
\caption{Computing time (in minutes) using \texttt{cardmatch} in the case study with 3, 5, and 10 levels of exposure to the earthquake.\label{match_time}}
\begin{adjustbox}{scale=0.7}
\begin{tabular}{ccc}
Exposure level & Sample size      & Time (min) \\ \hline 
1                                                             & 18208      & 0.05       \\
2                                                             & 70118      & 0.34       \\
3                                                             & 32953      & 0.08       \\ \hline
1                                                             & 22075      & 0.06       \\
2                                                             & 25977      & 0.06       \\
3                                                             & 24896      & 0.06       \\
4                                                             & 24279      & 0.06       \\
5                                                             & 24052      & 0.06       \\  \hline
1                                                             & 12084      & 0.03       \\
2                                                             & 9991       & 0.03       \\
3                                                             & 12513      & 0.04       \\
4                                                             & 13464      & 0.03       \\
5                                                             & 13119      & 0.03       \\
6                                                             & 11777      & 0.04       \\
7                                                             & 12813      & 0.03       \\
8                                                             & 11466      & 0.03       \\
9                                                             & 12071      & 0.03       \\
10                                                            & 11981      & 0.03      \\  \hline
\end{tabular}

\end{adjustbox}
\end{center}
\end{table}


\section{Results from the case study}
\label{sec:results}


\subsection{Assessing balance}

Table \ref{tab:balance} shows covariate balance for the matchings with 3, 5, and 10 levels of exposure to the earthquake.
In the template (target) sample there are 490 female and 510 male students, and the same is true across all exposure levels.
Due to space constraints, the table shows the counts for 3 covariates only (gender, school SES, and mother's education), but same pattern holds for all other 11 covariates (see Figure \ref{fig:balance} in the Supplementary Materials).
In other words, the marginal distributions of the 14 covariates are perfectly balanced relative to each other and to the representative template sample.
As a result, Figure \ref{fig:balance} shows that means of the indicators for the 78 covariate categories are perfectly balanced after matching.

\begin{table}[H]
\begin{center}
\caption{Distributional balance or fine balance across matched samples for 3, 5, and 10 levels of exposure to the earthquake.  Due to space constraints, the counts are shown for three covariates only but the same pattern holds for all other 11 covariates.\label{tab:balance}}
\subfloat[3 exposure levels]{
\begin{adjustbox}{scale=0.66}
\begin{tabular}{lccccc}
    & & & \multicolumn{3}{c}{Exposure level} \\
Covariate & Template & & 1 & 2 & 3 \\
\hline 
Gender &  & & & &    \\
\hspace{0.5cm}Male & 490 & & 490 & 490 & 490 \\ 
\hspace{0.5cm}Female & 510 & & 510 & 510 & 510 \\ 
School SES & & &  &  &   \\ 
\hspace{0.5cm}Low & 90 & & 90 & 90 & 90 \\ 
\hspace{0.5cm}Mid-low & 318 & & 318 & 318 & 318 \\ 
\hspace{0.5cm}Medium & 303 & & 303 & 303 & 303 \\ 
\hspace{0.5cm}Mid-high & 174 & & 174 & 174 & 174 \\ 
\hspace{0.5cm}High & 115 & & 115 & 115 & 115 \\ 
Mother's education \hspace{0.5cm} & & &  &  &    \\
\hspace{0.5cm}Primary & 321 & & 321 & 321 & 321 \\ 
\hspace{0.5cm}Secondary & 434 & & 434 & 434 & 434 \\ 
\hspace{0.5cm}Technical & 120 & & 120 & 120 & 120 \\ 
\hspace{0.5cm}College & 115 & & 115 & 115 & 115 \\ 
\hspace{0.5cm}Missing & 10 & & 10 & 10 & 10 \\ 
\hspace{1cm}\vdots & \multicolumn{5}{c}{} \\ \hline
\end{tabular}

\end{adjustbox}
}\\
\subfloat[5 exposure levels]{
\begin{adjustbox}{scale=0.66}
\begin{tabular}{lccccccc}
  & & & \multicolumn{5}{c}{Exposure level}\\
Covariate & Template & & 1 & 2 & 3 & 4 & 5\\
  \hline
Gender &  &  &  &  &  & &  \\ 
\hspace{0.5cm}Male & 490 & & 490 & 490 & 490 & 490 & 490 \\ 
\hspace{0.5cm}Female & 510 & & 510 & 510 & 510 & 510 & 510 \\ 
School SES & & & &  &  &  &    \\ 
\hspace{0.5cm}Low & 90 & & 90 & 90 & 90 & 90 & 90 \\ 
\hspace{0.5cm}Mid-low & 318 & & 318 & 318 & 318 & 318 & 318 \\ 
\hspace{0.5cm}Medium & 303 & & 303 & 303 & 303 & 303 & 303 \\ 
\hspace{0.5cm}Mid-high & 174 & & 174 & 174 & 174 & 174 & 174 \\ 
\hspace{0.5cm}High & 115 & & 115 & 115 & 115 & 115 & 115 \\ 
Mother's education\hspace{0.5cm} &  &  & & &  &  &    \\ 
\hspace{0.5cm}Primary & 321 & & 321 & 321 & 321 & 321 & 321 \\ 
\hspace{0.5cm}Secondary & 434 & & 434 & 434 & 434 & 434 & 434 \\ 
\hspace{0.5cm}Technical & 120 & & 120 & 120 & 120 & 120 & 120 \\ 
\hspace{0.5cm}College & 115 & & 115 & 115 & 115 & 115 & 115 \\ 
\hspace{0.5cm}Missing & 10 & & 10 & 10 & 10 & 10 & 10 \\ 
\hspace{1cm}\vdots & \multicolumn{7}{c}{} \\ \hline
\end{tabular}

\end{adjustbox}
}\\
\subfloat[10 exposure levels]{
\begin{adjustbox}{scale=0.66}
\begin{tabular}{lcccccccccccc}
  & & & \multicolumn{10}{c}{Exposure level} \\
Covariate & Template & & 1 & 2 & 3 & 4 & 5 & 6 & 7 & 8 & 9 & 10 \\
  \hline 
Gender &  &  &  &  &  &  &  &  &  &  & &  \\ 
\hspace{0.5cm}Male & 490 & & 490 & 490 & 490 & 490 & 490 & 490 & 490 & 490 & 490 & 490 \\ 
\hspace{0.5cm}Female & 510 & & 510 & 510 & 510 & 510 & 510 & 510 & 510 & 510 & 510 & 510 \\ 
School SES &  &  &  & & &  &  &  &  &  &  &    \\ 
\hspace{0.5cm}Low & 90 & & 90 & 90 & 90 & 90 & 90 & 90 & 90 & 90 & 90 & 90 \\ 
\hspace{0.5cm}Mid-low & 318 & & 318 & 318 & 318 & 318 & 318 & 318 & 318 & 318 & 318 & 318 \\ 
\hspace{0.5cm}Medium & 303 & & 303 & 303 & 303 & 303 & 303 & 303 & 303 & 303 & 303 & 303 \\ 
\hspace{0.5cm}Mid-high & 174 & & 174 & 174 & 174 & 174 & 174 & 174 & 174 & 174 & 174 & 174 \\ 
\hspace{0.5cm}High & 115 & & 115 & 115 & 115 & 115 & 115 & 115 & 115 & 115 & 115 & 115 \\ 
Mother's education\hspace{0.5cm} &  &  & & &  &  &  &  &  &  &  &    \\ 
\hspace{0.5cm}Primary & 321 & & 321 & 321 & 321 & 321 & 321 & 321 & 321 & 321 & 321 & 321 \\ 
\hspace{0.5cm}Secondary & 434 & & 434 & 434 & 434 & 434 & 434 & 434 & 434 & 434 & 434 & 434 \\ 
\hspace{0.5cm}Technical & 120 & & 120 & 120 & 120 & 120 & 120 & 120 & 120 & 120 & 120 & 120 \\ 
\hspace{0.5cm}College & 115 & & 115 & 115 & 115 & 115 & 115 & 115 & 115 & 115 & 115 & 115 \\ 
\hspace{0.5cm}Missing & 10 & & 10 & 10 & 10 & 10 & 10 & 10 & 10 & 10 & 10 & 10 \\ 
\hspace{1cm}\vdots & \multicolumn{12}{c}{} \\ \hline
\end{tabular}

\end{adjustbox}
}
\end{center}
\end{table}

\subsection{Visualizing effects}

One of the advantages of matching methods is that their adjustments are transparent, as illustrated in Table \ref{tab:balance}.
Also, the adjustments are made without looking at the outcomes, which aids the objectivity of the observational study \citep{rubin2008objective}.
Furthermore, the structure of the data after the adjustments is simple enough that we can analyze the effects by simply taking differences in means and even by plotting the outcomes.
This is illustrated in the Figure \ref{fig:boxplots}.

Relative to level 1, for each exposure level the figure shows the distribution of matched pair differences in outcomes for both school attendance (Figure \ref{fig:boxplots}(a)) and PSU scores (Figure \ref{fig:boxplots}(b)).
As expected, Figure \ref{fig:boxplots}(a) shows that as the exposure level increases, the impact on school attendance becomes more severe.
However, this pattern stands in stark contrast to the one in Figure \ref{fig:boxplots}(b), where we see no effect of the earthquake on university admission test scores.
This is striking, given the magnitude of the earthquake and its impact on school attendance that year.
Previous studies have documented a positive effect of school attendance on student achievement as measured by test scores (e.g., \citealt{lamdin1996evidence, gottfried2010evaluating, paredes2011should}), but this does not appear to be the case in the last year of school in Chile in terms of PSU scores.
In the following subsection we estimate the actual effects.

\begin{landscape}
\begin{figure}
\centering
\caption{Matched-pair differences in outcomes with respect to level 1}
\subfloat[Pair differences for attendance]{
\centering
\makebox{\includegraphics[width=0.65\textwidth]{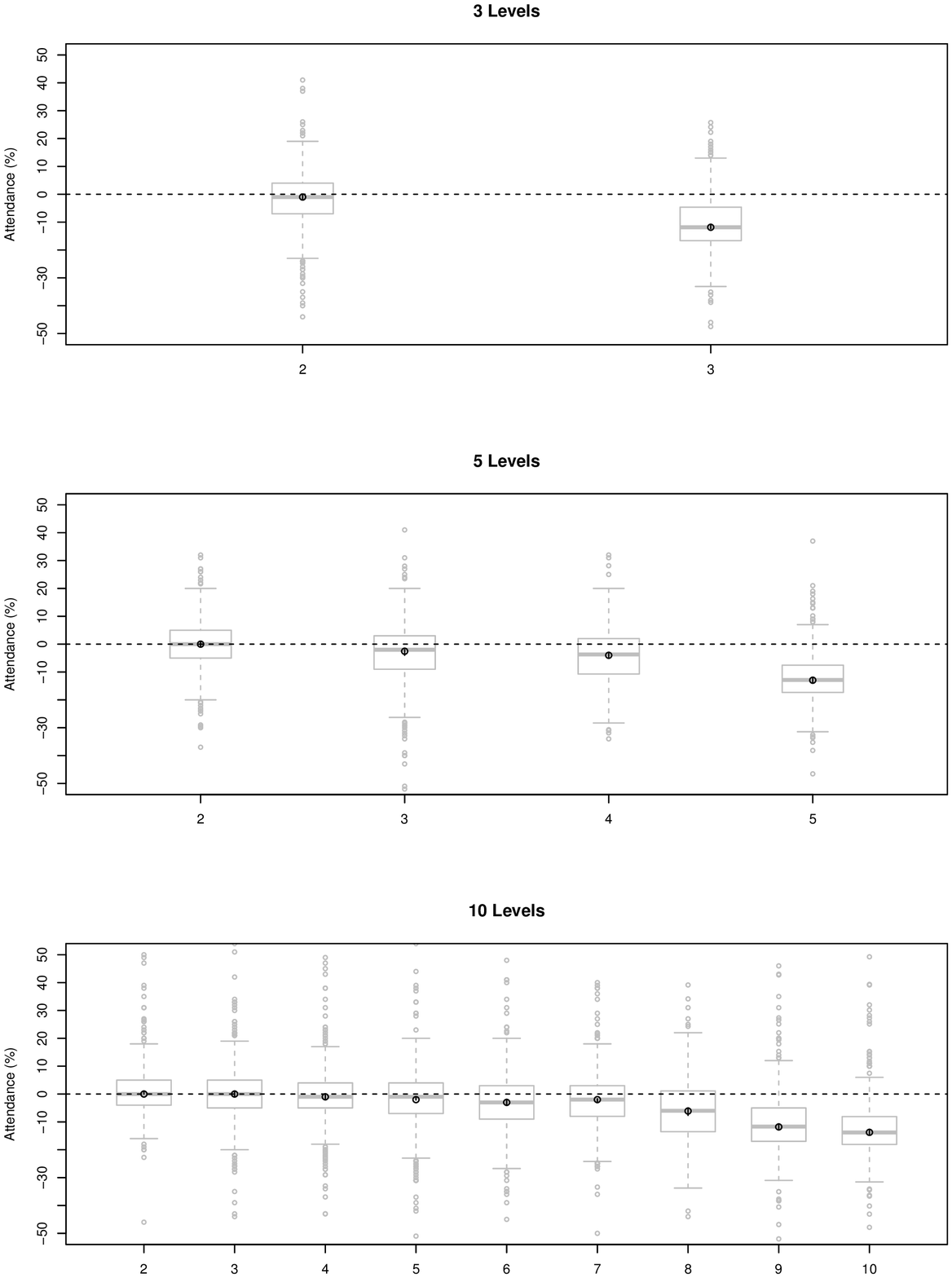}}
\label{fig:attendance}
}%
\subfloat[Pair differences for the PSU scores]{
\centering
\makebox{\includegraphics[width=0.65\textwidth]{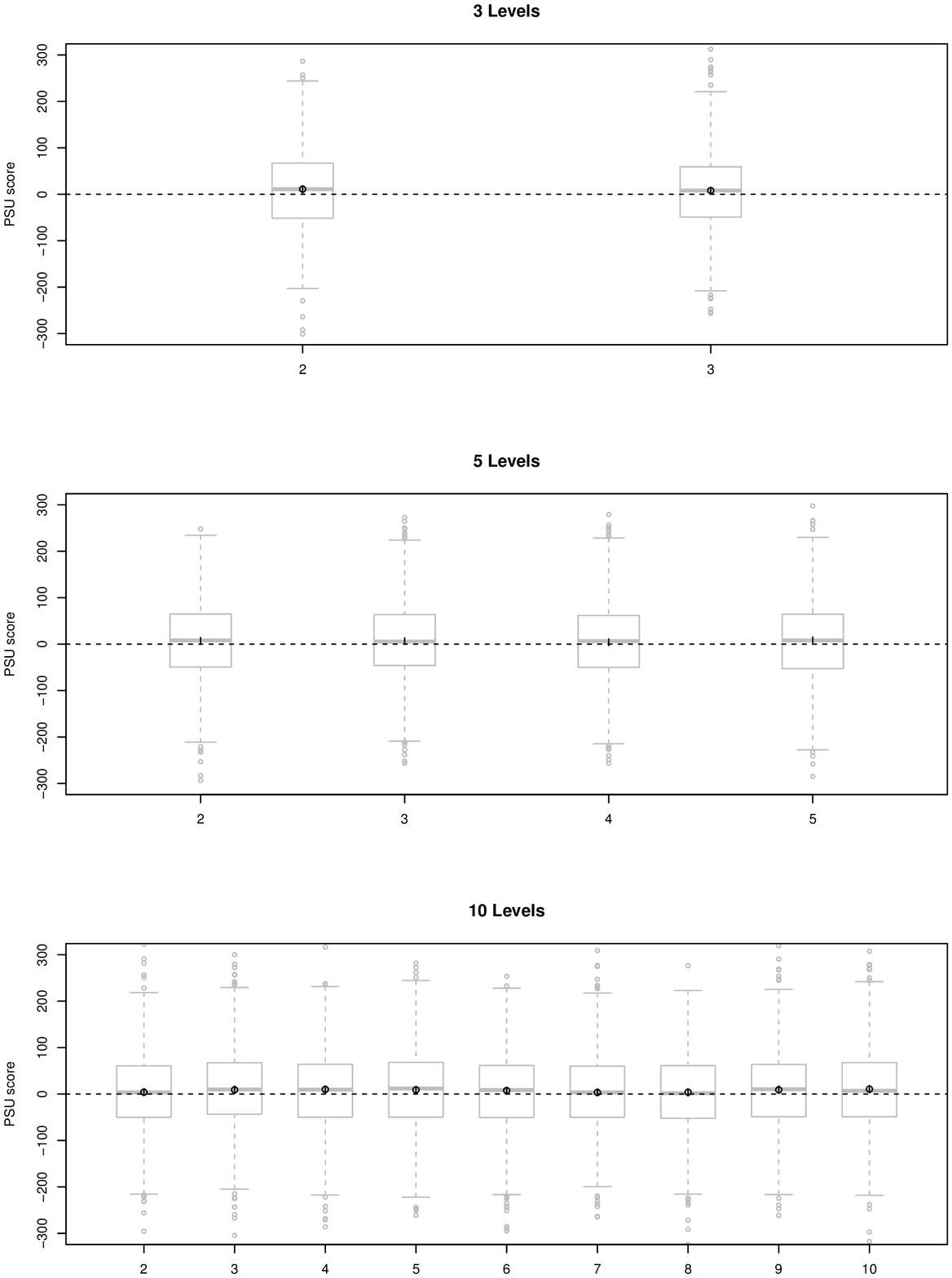}}
\label{fig:psu}
}
\label{fig:boxplots}
\end{figure}
\end{landscape}

\subsection{Estimating contrasts}

To estimate the effect of the earthquake, we contrast the outcomes in the matched samples under the different levels of exposure to the earthquake.
For this, we compute Hodges-Lehmann point estimates as well as 95\% confidence intervals following the procedure in \cite{hollander2015nonparametric}.
We compare the outcome of interest for each level $\mathcal{L} \in \mathbb{L}$ with respect to the control level $\mathcal{L}_1$.
The null hypothesis $H_0$ is defined as
\begin{equation*}
H_0: \ \tau^{i}_u = \tau^{i}_1 \ \mathrm{if} \ |R^i_u-R^i_1|<r^{i\ast}_{\alpha/2}
\end{equation*}
where $\tau^i_j$ is the treatment effect for level $j$ in $\mathcal{L}$, $R^i_j$ is the sum of the within-matched group ranks, and $r^{i\ast}_{\alpha/2}$ is a scalar such that, under the null hypothesis,
\begin{equation*}
P_0(|R^i_u-R^i_1|<r^{i\ast}_{\alpha/2},u=2,...,L) = 1-\alpha.
\end{equation*}
In this way, by setting an experiment-wise error rate of $\alpha$, we address the issue of multiple comparisons across different level levels.
We obtain the value of $r^{i\ast}_{\alpha/2}$ under the assumption that under the null hypothesis $H_0$ all $(L!)^n$ rank combinations are equally likely.

Table \ref{outcomes_cc} in Appendix E in the Supplementary Materials shows the effect estimates for the different exposure levels.
The results are consistent with the boxplots in Figure \ref{fig:boxplots}.
For example, in Table \ref{outcomes_cc}(c) with 10 exposure levels, the impact of the earthquake on school attendance increases with the exposure level, having an effect of 3 percentage points for level 6 and of almost 14 points for level 10, both relative to exposure level 1.
For exposure levels lower than 6, the effects on school attendance are not statistically significant.
In the second column of the table, it is again surprising to see that despite this negative and significant impact on school attendance, the earthquake did not have a significant impact on the PSU test scores, a test that was taken at the end of that school year.

To assess the sensitivity of these findings to hidden biases, we compute Rosenbaum (\citeyear{rosenbaum1987sensitivity,rosenbaum2002observational,rosenbaum2015two}) bounds.
These bounds quantify the magnitude that an unobserved covariate would need to have in order to explain away a significant effect and qualitatively change its interpretation.
This magnitude is summarized by the parameter $\Gamma_c$ in Table \ref{sensitivity}.
We see that for high exposure levels the estimates are quite insensitive to hidden biases.
For example, with 10 exposure levels (Table \ref{sensitivity}(c)) in order to explain away the estimated effect of exposure level 10 relative to level 1 on school attendance, an unobserved covariate that is perfectly associated with this outcome would need to increase the odds of being exposed to level 10 as opposed to level 1 by a factor of 28.
For exposure levels 9 and 10, these values are 13.6 and 3.8, respectively.
By way of contrast, Hammond's (\citeyear{hammond1964smoking}) classical study on the effect of smoking on lung cancer becomes sensitive at $\Gamma_c \approx 6$ (\citealt{rosenbaum2002observational}; Section 4.3.2).

\begin{table}[H]
\begin{center}
\caption{Critical $\Gamma$ ($\Gamma_c$) for Sensitivity Analysis\label{sensitivity}}
\subfloat[3 levels]{
\begin{adjustbox}{scale=0.7}
\begin{tabular}{lc}
Level of exposure (compared to level 1) & $\Gamma_c$ \\ 
  \hline
2 & 1.34 \\ 
  3 & 10.68 \\ 
   \hline
\end{tabular}

\end{adjustbox}
}\\
\subfloat[5 levels]{
\begin{adjustbox}{scale=0.7}
\begin{tabular}{lc}
Level of exposure (compared to level 1) & $\Gamma_c$\\ 
  \hline
3 & 1.82  \\ 
  4 & 2.58 \\ 
  5 & 17.64 \\ 
   \hline
\end{tabular}

\end{adjustbox}
}\\
\subfloat[10 levels]{
\begin{adjustbox}{scale=0.7}
\begin{tabular}{lc}
Level of exposure (compared to level 1) & $\Gamma_c$  \\ 
  \hline
5 & 1.36  \\ 
  6 & 1.95 \\ 
  7 & 1.72 \\ 
  8 & 3.48 \\ 
  9 & 9.49  \\ 
  10 & 18.47 \\ 
   \hline
\end{tabular}

\end{adjustbox}
}
\end{center}
\end{table}

\section{Summary and remarks}
\label{sec:summary}

In this paper, we have proposed a new approach to address three challenges in matching in observational studies.
The first challenge relates to handling multi-valued treatments (possibly tens or hundreds of them, either ordered or un-ordered) without estimating the generalized propensity score, directly and flexibly balancing the observed covariates, and facilitating transparent analysis for the outcomes, such as graphical displays.
The second challenge relates to building matched samples that are not only balanced but that are representative of a population of interest, in such as way that we obtain ``representative estimates'' of target causal effects.
Arguably, this second challenge goes beyond matching and also applies to regression as a method for covariate adjustment for causal inference (see \citealt{aronow2016does}).
The third challenge relates to matching in larger data sets than usually considered, with hundred of thousands of observations.

To overcome these challenges, instead of simultaneously matching across treatments groups, we separately match each treatment group to a representative random sample of the population of interest.
For this, we impose exact distributional or fine balance constraints.
This guarantees that the marginal distributions of the matched samples across treatment groups will be identical to each other and to the template in terms of the observed covariates.
The effectiveness of the approach relies on the use of a linear-sized MIP formulation, which we show is (i) as strong as much larger quadratic-sized formulations, and (ii) is integral (and hence polynomially solvable) when only two covariates or nested covariates are considered.
This integrality property is not preserved for more general covariate structures, but the formulation still retains its  practical effectiveness in such settings.

We have used this new matching approach to estimate the impact of an earthquake on the educational achievement of high school students.
In particular, we estimated the effect of levels of exposure to the 2010 Chilean earthquake on both school attendance and university admission test scores.
We documented negatively increasing effects of the strength of the earthquake on school attendance, but no effect on university admission test scores.

\onehalfspacing
\bibliography{mybibliography11,mybibliography_extra}
\bibliographystyle{asa}

\pagebreak
\setcounter{page}{1}
\section*{Online Supplementary Materials}

\subsection*{Appendix A: Earthquake intensity levels}

Using the estimated PGA, we created three measures of exposure to the earthquake: one with three levels, another with five levels, and a final one with ten levels.

We defined the exposure with three intensity levels as follows.
\begin{itemize}
\item[$-$] Low PGA (PGA$<0.08$): felt by many people indoors; no buildings received damage.
\item[$-$] Medium PGA ($0.08\leq$ PGA $\leq 0.25$): felt by most or all people indoors; some people were frightened; damages in some (non-resistant) buildings.
\item[$-$] High PGA (PGA$>0.25$): many people were frightened; severe shaking; damages in resistant buildings.
\end{itemize}

For the versions of the exposure with five and ten levels, we divided the students into PGA quintiles and deciles, respectively.

\subsection*{Appendix B: Covariate profiles}

\begin{table}[H]
\begin{center}
\caption{\label{template_comparison}Covariate profile of the population and template sample}
\begin{adjustbox}{scale=0.67}
\begin{tabular}{lcc}
\hline
 Covariate & Population & Template \\ 
  \hline
Female & 0.54 & 0.51 \\
Indigenous & & \\  
\hspace{0.5cm}Indigenous  & 0.08 & 0.07 \\  
\hspace{0.5cm}Missing & 0.15 & 0.14 \\
Father's education & & \\
\hspace{0.5cm}Secondary & 0.39 & 0.40 \\ 
\hspace{0.5cm}Technical & 0.09 & 0.09 \\ 
\hspace{0.5cm}College & 0.15 & 0.14 \\ 
\hspace{0.5cm}Missing & 0.05 & 0.04 \\
Mother's education & & \\
\hspace{0.5cm}Secondary & 0.41 & 0.43 \\ 
\hspace{0.5cm}Technical & 0.13 & 0.12 \\ 
\hspace{0.5cm}College & 0.12 & 0.12 \\ 
\hspace{0.5cm}Missing & 0.01 & 0.01 \\
Household income (2008 CL\$1000) & & \\ 
\hspace{0.5cm}100-200 & 0.26 & 0.26 \\ 
\hspace{0.5cm}200-400 & 0.30 & 0.31 \\ 
\hspace{0.5cm}400-600 & 0.13 & 0.12 \\ 
\hspace{0.5cm}600-1400 & 0.13 & 0.14 \\ 
\hspace{0.5cm}1400 or more & 0.09 & 0.09 \\ 
\hspace{0.5cm}Missing & 0.01 & 0.02 \\ 
Number of books at home & & \\
\hspace{0.5cm}1-10 & 0.19 & 0.19 \\ 
\hspace{0.5cm}11-50 & 0.46 & 0.46 \\ 
\hspace{0.5cm}51-10 & 0.16 & 0.16 \\ 
\hspace{0.5cm}More than 100 & 0.16 & 0.16 \\ 
\hspace{0.5cm}Missing & 0.01 & 0.01 \\  
\hline
\end{tabular}

\end{adjustbox}
\end{center}
\end{table}

\begin{table}[H]
\begin{center}
\caption*{Table \ref{template_comparison} (continued): Covariate profile of the population and template sample}
\begin{adjustbox}{scale=0.67}
\begin{tabular}{lcc}
 Covariate & Population & Template \\ 
  \hline
 Student's attendance (deciles) & & \\ 
\hspace{0.5cm}2 & 0.12 & 0.11 \\ 
\hspace{0.5cm}3 & 0.08 & 0.09 \\ 
\hspace{0.5cm}4 & 0.10 & 0.10 \\ 
\hspace{0.5cm}5 & 0.13 & 0.13 \\ 
\hspace{0.5cm}6 & 0.08 & 0.08 \\ 
\hspace{0.5cm}7 & 0.09 & 0.09 \\ 
\hspace{0.5cm}8 & 0.10 & 0.09 \\ 
\hspace{0.5cm}9 & 0.11 & 0.10 \\ 
\hspace{0.5cm}10 & 0.15 & 0.16 \\
Student's GPA 2008 (deciles) & & \\
\hspace{0.5cm}2 & 0.11 & 0.11 \\ 
\hspace{0.5cm}3 & 0.11 & 0.09 \\ 
\hspace{0.5cm}4 & 0.10 & 0.11 \\ 
\hspace{0.5cm}5 & 0.10 & 0.10 \\ 
\hspace{0.5cm}6 & 0.10 & 0.11 \\ 
\hspace{0.5cm}7 & 0.09 & 0.09 \\ 
\hspace{0.5cm}8 & 0.10 & 0.10 \\ 
\hspace{0.5cm}9 & 0.09 & 0.10 \\ 
\hspace{0.5cm}10 & 0.08 & 0.08 \\ 
Student's test scores (deciles) & & \\ 
\hspace{0.5cm}2 & 0.08 & 0.07 \\ 
\hspace{0.5cm}3 & 0.09 & 0.10 \\ 
\hspace{0.5cm}4 & 0.09 & 0.10 \\ 
\hspace{0.5cm}5 & 0.10 & 0.10 \\ 
\hspace{0.5cm}6 & 0.10 & 0.11 \\ 
\hspace{0.5cm}7 & 0.11 & 0.10 \\ 
\hspace{0.5cm}8 & 0.12 & 0.12 \\ 
\hspace{0.5cm}9 & 0.12 & 0.11 \\ 
\hspace{0.5cm}10 & 0.12 & 0.13 \\ 
\hspace{0.5cm}Missing & 0.01 & 0.01 \\
School administration & & \\  
\hspace{0.5cm}Public & 0.34 & 0.34 \\
\hspace{0.5cm}Private subsidized (voucher) & 0.55 & 0.55 \\ 
Rural school & 0.03 & 0.02 \\ 
Catholic school & 0.24 & 0.23 \\ 
School SES & & \\
\hspace{0.5cm}Mid-low & 0.32 & 0.32 \\ 
\hspace{0.5cm}Medium & 0.29 & 0.30 \\ 
\hspace{0.5cm}Mid-high & 0.18 & 0.17 \\ 
\hspace{0.5cm}High & 0.11 & 0.12 \\ 
School's test scores (deciles) & & \\  
\hspace{0.5cm}2 & 0.07 & 0.07 \\ 
\hspace{0.5cm}3 & 0.09 & 0.07 \\ 
\hspace{0.5cm}4 & 0.09 & 0.09 \\ 
\hspace{0.5cm}5 & 0.10 & 0.11 \\ 
\hspace{0.5cm}6 & 0.11 & 0.12 \\ 
\hspace{0.5cm}7 & 0.10 & 0.11 \\ 
\hspace{0.5cm}8 & 0.12 & 0.12 \\ 
\hspace{0.5cm}9 & 0.12 & 0.11 \\ 
\hspace{0.5cm}10 & 0.13 & 0.13 \\ 
   \hline
\end{tabular}
\end{adjustbox}
\end{center}
\end{table}

\subsection*{Appendix C: Proofs}

To prove Proposition ~\ref{allcovariatesprop} we use the following lemma.

\begin{lemma}\label{twocovlemma}
If $P\leq 2$  and $\bra{\boldsymbol{z},\;\boldsymbol{v}}$ is an non-integral feasible solution for the LP relaxation of \eqref{original}, then one of the following conditions holds:
\begin{enumerate}
  \item There exist $\underline{\boldsymbol{z}}$ and $\overline{\boldsymbol{z}}$ such that
  \begin{itemize}\label{firstcaselemma}
    \item $\sum\nolimits_{\ell \in \mathcal{L}_{p,k}}   \underline{z}_\ell=\sum\nolimits_{\ell \in \mathcal{L}_{p,k}}   \overline{z}_\ell=\sum\nolimits_{\ell \in \mathcal{L}_{p,k}}   {z}_\ell  $ for all $p\in \mathcal{P}$ and $k\in\mathcal{K}(p)$,
    \item $  \underline{z}_\ell, \overline{z}_\ell\in [0,1]$ for all $\ell \in \mathcal{L}$,
    \item $\underline{\boldsymbol{z}}\neq\overline{\boldsymbol{z}}$, and
    \item $\boldsymbol{z}=(1/2)\underline{\boldsymbol{z}}+(1/2)\overline{\boldsymbol{z}}$.
  \end{itemize}
  \item There exist $q_1,q_2\in \mathcal{P}$, $u_1\in\mathcal{K}(q_1)$, $u_2\in\mathcal{K}(q_2)$, $\bra{\underline{\boldsymbol{z}},\underline{\boldsymbol{v}}}$ and $\bra{\overline{\boldsymbol{z}},\overline{\boldsymbol{v}}}$ such that
  \begin{itemize}\label{secondcaselemma}
    \item $\sum\nolimits_{\ell \in \mathcal{L}_{p,k}}   \underline{z}_\ell=\sum\nolimits_{\ell \in \mathcal{L}_{p,k}}   \overline{z}_\ell=\sum\nolimits_{\ell \in \mathcal{L}_{p,k}}   {z}_\ell  $ for all $p\in \mathcal{P}$ and $k\in\mathcal{K}(p)$ such that $\bra{p,k}\notin \set{\bra{q_1,u_1}, \bra{q_2,u_2}}$,
    \item $\abs{ \sum\nolimits_{\ell \in \mathcal{L}_{p,k}}   \underline{z}_\ell - N_{p,k}}   = \underline{v}_{p,k} $ and $\abs{ \sum\nolimits_{\ell \in \mathcal{L}_{p,k}}   \overline{z}_\ell - N_{p,k}}   = \overline{v}_{p,k} $ for all $\bra{p,k}\in \set{\bra{q_1,u_1}, \bra{q_2,u_2}}$,
  \item $\bra{\underline{\boldsymbol{z}},\underline{\boldsymbol{v}}}\neq\bra{\overline{\boldsymbol{z}},\overline{\boldsymbol{v}}}$, and
    \item $\bra{\boldsymbol{z},\boldsymbol{v}}=\frac{1}{2}\bra{\underline{\boldsymbol{z}},\underline{\boldsymbol{v}}}+\frac{1}{2}\bra{\overline{\boldsymbol{z}},\overline{\boldsymbol{v}}}$.
  \end{itemize}
\end{enumerate}
\end{lemma}
\begin{proof}[Proof]
  Note that $N_{p,k}\in \mathbb{Z}$ and $\bra{\boldsymbol{z},\;\boldsymbol{v}}$ satisfies \eqref{prjlpaa}. If there exist $p\in \mathcal{P}$ and $k\in\mathcal{K}(p)$ for which $\bra{\boldsymbol{z},\;\boldsymbol{v}}$ satisfies \eqref{prjlpaa} as a strict inequality, the let $\varepsilon = v_{p,k} - \abs{ \sum\nolimits_{\ell \in \mathcal{L}_{p,k}}   z_\ell - N_{p,k}}$, $\bra{\underline{\boldsymbol{z}},\underline{\boldsymbol{v}}}=\bra{\overline{\boldsymbol{z}},\overline{\boldsymbol{v}}}=\bra{\boldsymbol{z},\;\boldsymbol{v}}$ and change $\underline{v}_{p,k}$ to ${v}_{p,k}-\varepsilon$ and $\overline{v}_{p,k}$ to ${v}_{p,k}+\varepsilon$ to obtain case \ref{secondcaselemma}.

  If all inequalities \eqref{prjlpaa} are satisfied at equality by $\bra{\boldsymbol{z},\;\boldsymbol{v}}$, then for each $p\in \mathcal{P}$ and $k\in\mathcal{K}(p)$ we have that either $v_{p,k}\notin \mathbb{Z}_+$ or $\abs{\set{\ell \in \mathcal{L}_{p,k}\,:\, z_\ell\in (0,1)}}\geq 2$.
  Hence, we may construct a sequences $\set{s_1,\ldots,s_S}\subseteq \set{1,\ldots,L}$  and $\set{r_1,\ldots,r_{S+1}}\subseteq \bigcup_{p\in \mathcal{P}}  \mathcal{K}(p)$ such that (without loss of generality by possibly interchanging $p_1$ and $p_2$):
  \begin{itemize}
  	\item $s_j\neq s_l$ and $r_j\neq r_S$ for all $j,l\in \set{1,\ldots,S}$ with $j\neq l$,

  \item $z_{\ell_{s_j}}\in (0,1)$ for all $j\in \set{1,\ldots,S}$,
  	\item $r_j\in \mathcal{K}\bra{p_{h(j)}}$ for all $j\in \set{1,\ldots,S+1}$
  	where
  \[h(j)=2-j+2\left\lfloor j/2\right\rfloor=\begin{cases}1&\text{$j$ is odd}\\ 2& \text{$j$ is even},\end{cases} \]
  \item $\ell_{s_j}\in \mathcal{L}_{p_{h(j)},r_j}$ and $\ell_{s_j}\in \mathcal{L}_{p_{h(j+1)},r_{j+1} }$ for all $j\in \set{1,\ldots,S}$,
  \end{itemize}
  and either
  \begin{equation}\label{firstconditionLP}
  \text{ $S$ is even, $h(1)=h(S+1)$ and $r_1=r_{S+1}$ }
  \end{equation}
  or
  \begin{subequations}\label{secondconditionLP}
  \begin{align}
  v_{{p_{h(1)},r_1}},v_{{p_{h(S+1)},r_{S+1} }}&\in (0,1),\\
  \set{\ell \in \mathcal{L}_{{p_{h(1)},r_1}}\,:\, z_\ell\in (0,1)}&=\set{z_{\ell_{s_1}}},\\
  \set{\ell \in \mathcal{L}_{{p_{h(S+1)},r_{S+1} }}\,:\, z_\ell\in (0,1)}&=\set{z_{\ell_{s_S}}}.
  \end{align}
  \end{subequations}

  If \eqref{firstconditionLP} holds let  $\varepsilon =\min_{i=1}^S\set{z_{\ell_{s_j}},1-z_{\ell_{s_j}}}$, $\underline{\boldsymbol{z}}$ and $\overline{\boldsymbol{z}}$ be such that
  \begin{align*}
  \underline{z}_\ell &=
  \begin{cases}
   z_\ell -\varepsilon & \text{if $\ell=s_j$ for an odd $j\in S$,}\\
   z_\ell +\varepsilon & \text{if $\ell=s_j$ for an even $j\in S$,}\\
   z_\ell &\text{otherwise}
  \end{cases}\\
  \overline{z}_\ell &=
  \begin{cases}
   z_\ell +\varepsilon & \text{if $\ell=s_j$ for an odd $j\in S$,}\\
   z_\ell -\varepsilon & \text{if $\ell=s_j$ for an even $j\in S$,}\\
   z_\ell &\text{otherwise},
  \end{cases}
  \end{align*}
  to obtain case \ref{firstcaselemma}.

  If instead \eqref{secondconditionLP} holds let $q_1=p_{h(1)}$, $q_2=p_{h(S+1)}$, $u_1=r_1$, $u_2=r_{S+1}$, \[\varepsilon =\min\set{\min_{i=1}^S\set{z_{\ell_{s_j}},1-z_{\ell_{s_j}}},v_{{p_{h(1)},r_1}},v_{{p_{h(S+1)},r_{S+1} }}}\]  $\underline{\boldsymbol{z}}$ and $\overline{\boldsymbol{z}}$ be as defined previously with this new $\varepsilon$. Then, let $\underline{\boldsymbol{v}}$ and $\overline{\boldsymbol{v}}$ be such that
\begin{align*}
\underline{v}_{p,k} &=
\begin{cases}
 {v}_{p,k} +\varepsilon & \text{if $p=p_{h(1)}$, $k={r_1}$ and $N_{p,k}-\sum_{\ell \in \mathcal{L}_{p,k}}   z_\ell>0$}\\
  {v}_{p,k} -\varepsilon & \text{if $p=p_{h(1)}$, $k={r_1}$ and $N_{p,k}-\sum_{\ell \in \mathcal{L}_{p,k}}   z_\ell< 0$}\\
   {v}_{p,k} -\varepsilon & \text{if $p=p_{h(S+1)}$, $k={r_{S+1}}$, $S$ is even and $N_{p,k}-\sum_{\ell \in \mathcal{L}_{p,k}}   z_\ell>0$}\\
  {v}_{p,k} +\varepsilon & \text{if $p=p_{h(S+1)}$, $k={r_{S+1}}$, $S$ is even and $N_{p,k}-\sum_{\ell \in \mathcal{L}_{p,k}}   z_\ell <0$}\\
  {v}_{p,k} +\varepsilon & \text{if $p=p_{h(S+1)}$, $k={r_{S+1}}$, $S$ is odd and $N_{p,k}-\sum_{\ell \in \mathcal{L}_{p,k}}   z_\ell>0$}\\
{v}_{p,k} -\varepsilon & \text{if $p=p_{h(S+1)}$, $k={r_{S+1}}$, $S$ is odd and $N_{p,k}-\sum_{\ell \in \mathcal{L}_{p,k}}   z_\ell< 0$}\\
 {v}_{p,k} & \text{otherwise}
\end{cases}\\
\overline{v}_{p,k} &=
\begin{cases}
 {v}_{p,k} -\varepsilon & \text{if $p=p_{h(1)}$, $k={r_1}$ and $N_{p,k}-\sum_{\ell \in \mathcal{L}_{p,k}}   z_\ell>0$}\\
  {v}_{p,k} +\varepsilon & \text{if $p=p_{h(1)}$, $k={r_1}$ and $N_{p,k}-\sum_{\ell \in \mathcal{L}_{p,k}}   z_\ell<0$}\\
   {v}_{p,k} +\varepsilon & \text{if $p=p_{h(S+1)}$, $k={r_{S+1}}$, $S$ is even and $N_{p,k}-\sum_{\ell \in \mathcal{L}_{p,k}}   z_\ell>0$}\\
  {v}_{p,k} -\varepsilon & \text{if $p=p_{h(S+1)}$, $k={r_{S+1}}$, $S$ is even and $N_{p,k}-\sum_{\ell \in \mathcal{L}_{p,k}}   z_\ell< 0$}\\
  {v}_{p,k} -\varepsilon & \text{if $p=p_{h(S+1)}$, $k={r_{S+1}}$, $S$ is odd and $N_{p,k}-\sum_{\ell \in \mathcal{L}_{p,k}}   z_\ell>0$}\\
{v}_{p,k} +\varepsilon & \text{if $p=p_{h(S+1)}$, $k={r_{S+1}}$, $S$ is odd and $N_{p,k}-\sum_{\ell \in \mathcal{L}_{p,k}}   z_\ell< 0$}\\
 {v}_{p,k} & \text{otherwise}
\end{cases}
\end{align*}
  to obtain case \ref{secondcaselemma}.
\end{proof}

\begin{proof}[Proof of Proposition ~\ref{allcovariatesprop}]

For case \ref{allcovariatespropcase1} we prove the result for $P=2$. The result follows for $P=1$ by copying the covariate as the result for $P=2$ does not require that the covariates are different. Furthermore, because of Proposition~\ref{projectionprop} it suffices to prove the result for \eqref{projected}.

Assume for a contradiction that the LP relaxation of \eqref{projected} has a non-integral extreme point $\bra{\boldsymbol{z},\;\boldsymbol{v}}$. Then by Lemma~\ref{twocovlemma} there exist $\bra{\underline{\boldsymbol{z}},\underline{\boldsymbol{v}}}$ and $\bra{\overline{\boldsymbol{z}},\overline{\boldsymbol{v}}}$ that are feasible for the LP relaxation of \eqref{projected},  $\bra{\underline{\boldsymbol{z}},\underline{\boldsymbol{v}}}\neq\bra{\overline{\boldsymbol{z}},\overline{\boldsymbol{v}}}$ and $\bra{\boldsymbol{z},\;\boldsymbol{v}}=\frac{1}{2}\bra{\underline{\boldsymbol{z}},\underline{\boldsymbol{v}}}+\frac{1}{2}\bra{\overline{\boldsymbol{z}},\overline{\boldsymbol{v}}}$, which contradicts $\bra{\boldsymbol{z},\boldsymbol{v}}$ being an extreme point.

For case \ref{allcovariatespropcase2} we only need to prove the result for $P\geq 3$. Again, assume for a contradiction that the LP relaxation of \eqref{projected} has a non-integral extreme point $\bra{\boldsymbol{z},\;\boldsymbol{v}}$. If we restrict this $\boldsymbol{v}$ to the last two covariates to we can apply Lemma~\ref{twocovlemma} to obtain $\bra{\underline{\boldsymbol{z}},\underline{\boldsymbol{v}}}$ and $\bra{\overline{\boldsymbol{z}},\overline{\boldsymbol{v}}}$ (with $\underline{\boldsymbol{v}}/\overline{\boldsymbol{v}}$ restricted to the last two covariates) such that    $\bra{\underline{\boldsymbol{z}},\underline{\boldsymbol{v}}}\neq\bra{\overline{\boldsymbol{z}},\overline{\boldsymbol{v}}}$ and $\bra{\boldsymbol{z},\;\boldsymbol{v}}=\frac{1}{2}\bra{\underline{\boldsymbol{z}},\underline{\boldsymbol{v}}}+\frac{1}{2}\bra{\overline{\boldsymbol{z}},\overline{\boldsymbol{v}}}$ (with $\boldsymbol{v}/\underline{\boldsymbol{v}}/\overline{\boldsymbol{v}}$ restricted to the last two covariates). Furthermore, $\bra{\underline{\boldsymbol{z}},\underline{\boldsymbol{v}}}$ and $\bra{\overline{\boldsymbol{z}},\overline{\boldsymbol{v}}}$ satisfy all the constraints of LP relaxation of \eqref{projected} except for constraints  \eqref{prjlpaa} associated to the first $P-2$ covariates.

If $\bra{\underline{\boldsymbol{z}},\underline{\boldsymbol{v}}}$ and $\bra{\overline{\boldsymbol{z}},\overline{\boldsymbol{v}}}$  came from  case \ref{firstcaselemma} Lemma~\ref{twocovlemma} we  extend $\underline{\boldsymbol{v}}/\overline{\boldsymbol{v}}$  to all covariates by letting $\underline{{v}}_{p,k}=\overline{{v}}_{p,k}=v_{p,k}$ for all $p\in \set{p_1,\ldots,p_{P-2}}$ and $k\in \mathcal{K}\bra{p}$.
This extension is feasible for the complete  LP relaxation of \eqref{projected} because
$\sum\nolimits_{\ell \in \mathcal{L}_{p_{P-1},k}}   \underline{z}_\ell=\sum\nolimits_{\ell \in \mathcal{L}_{p_{P-1},k}}   \overline{z}_\ell=\sum\nolimits_{\ell \in \mathcal{L}_{p_{P-1},k}}   {z}_\ell  $
for all  $k\in\mathcal{K}(p_{P-1})$, and because by assumption for all $p\in \set{p_{P-1},p_{P}}$, $k\in\mathcal{K}(p)$ and
$i\in \set{1,\ldots,P-2}$ there exist $k'\in  \mathcal{K}\bra{p_i}$ such that
 $\mathcal{L}_{p,k}\subseteq \mathcal{L}_{p_i,k'}$.

 If $\bra{\underline{\boldsymbol{z}},\underline{\boldsymbol{v}}}$ and $\bra{\overline{\boldsymbol{z}},\overline{\boldsymbol{v}}}$  came from  case \ref{secondcaselemma} Lemma~\ref{twocovlemma}, let $q_1,q_2\in \set{p_{P-1},p_{P}}$, $u_1\in\mathcal{K}(q_1)$, $u_2\in\mathcal{K}(q_2)$ from the Lemma be such that $\abs{ \sum\nolimits_{\ell \in \mathcal{L}_{p,k}}   \underline{z}_\ell - N_{p,k}}   = \underline{v}_{p,k} $ and $\abs{ \sum\nolimits_{\ell \in \mathcal{L}_{p,k}}   \overline{z}_\ell - N_{p,k}}   = \overline{v}_{p,k} $ for all $\bra{p,k}\in \set{\bra{q_1,u_1}, \bra{q_2,u_2}}$. Then if we extend $\underline{\boldsymbol{v}}/\overline{\boldsymbol{v}}$  be letting
 \begin{align*}
\underline{v}_{p,k} = \overline{v}_{p,k}=
\begin{cases}
 {v}_{q_1,u_1}  & \text{if $k\in  \mathcal{K}\bra{p}$ is such that $\mathcal{L}_{q_1,u_1}\subseteq \mathcal{L}_{p,k}$}\\
 {v}_{q_2,u_2}  & \text{if $k\in  \mathcal{K}\bra{p}$ is such that $\mathcal{L}_{q_2,u_1}\subseteq \mathcal{L}_{p,k}$}\\
 {v}_{p,k} & \text{otherwise}
\end{cases}
\end{align*}
we again obtain an extension that is feasible for the complete  LP relaxation of \eqref{projected}.

In both cases we end up with a contradiction of $\bra{\boldsymbol{z},\boldsymbol{v}}$ being an extreme point.
\end{proof}

\begin{proof}[Proof of Lemma ~\ref{examplelemma}]
Let
  \begin{alignat*}{3}
    \mathbf{x}_1&=\bra{k_1,k_1,k_1},&\quad & \mathbf{x}_2&=\bra{k_3,k_3,k_3}, \\
    \mathbf{x}_1&=\bra{k_1,k_2,k_3},&\quad & \mathbf{x}_2&=\bra{k_3,k_2,k_1},\\
    \mathbf{x}_1&=\bra{k_2,k_1,k_2},&\quad & \mathbf{x}_2&=\bra{k_2,k_3,k_2}, \\
  \end{alignat*}
$T=3$, $L=6$ and $N_{p,k}=1$ for all $p$ and $k$. The feasible region of the LP relaxation of \eqref{projected} for this case is given by
\begin{alignat*}{7}
    \abs{ z_{\ell_1} +z_{\ell_3} -1 } &\leq v_{p_1,k_1}, &\quad   \abs{ z_{\ell_5} +z_{\ell_6} -1 } &\leq v_{p_1,k_2}, &\quad &   \abs{ z_{\ell_2} +z_{\ell_4} -1 } &\leq v_{p_1,k_3},\\
    \abs{ z_{\ell_1} +z_{\ell_5} -1 } &\leq v_{p_2,k_1}, &\quad   \abs{ z_{\ell_3} +z_{\ell_4} -1 } &\leq v_{p_2,k_2}, &\quad &   \abs{ z_{\ell_2} +z_{\ell_6} -1 } &\leq v_{p_2,k_3}\\
    \abs{ z_{\ell_1} +z_{\ell_4} -1 } &\leq v_{p_3,k_1}, &\quad    \abs{ z_{\ell_5} +z_{\ell_6} -1 } &\leq v_{p_3,k_2}, &\quad &    \abs{ z_{\ell_2} +z_{\ell_3} -1 } &\leq v_{p_3,k_3}\\
    \sum_{i=1}^6 z_{\ell_i}           &=1,               &\quad                      0\leq z_{\ell_i}&\leq 1             &\quad&\forall i\in \set{1,\ldots,6}.&
 \end{alignat*}
Using CDDLib \citep{fukuda2001cddlib} we can check that this LP has $11$ fractional extreme points out of a total of $31$. In particular, $z_{\ell_i}=1/2$ for all $i\in\set{1,\ldots,6}$ and $v_{p,k}=0$ for all $p$ and $k$ is one such fractional extreme point.
\end{proof}

\subsection*{Appendix D: Covariate balance}

\begin{figure}[H]
\centering
\caption{Standardized differences in means in covariates before and after matching for 10 levels of exposure.}
\makebox{\includegraphics[scale=.85]{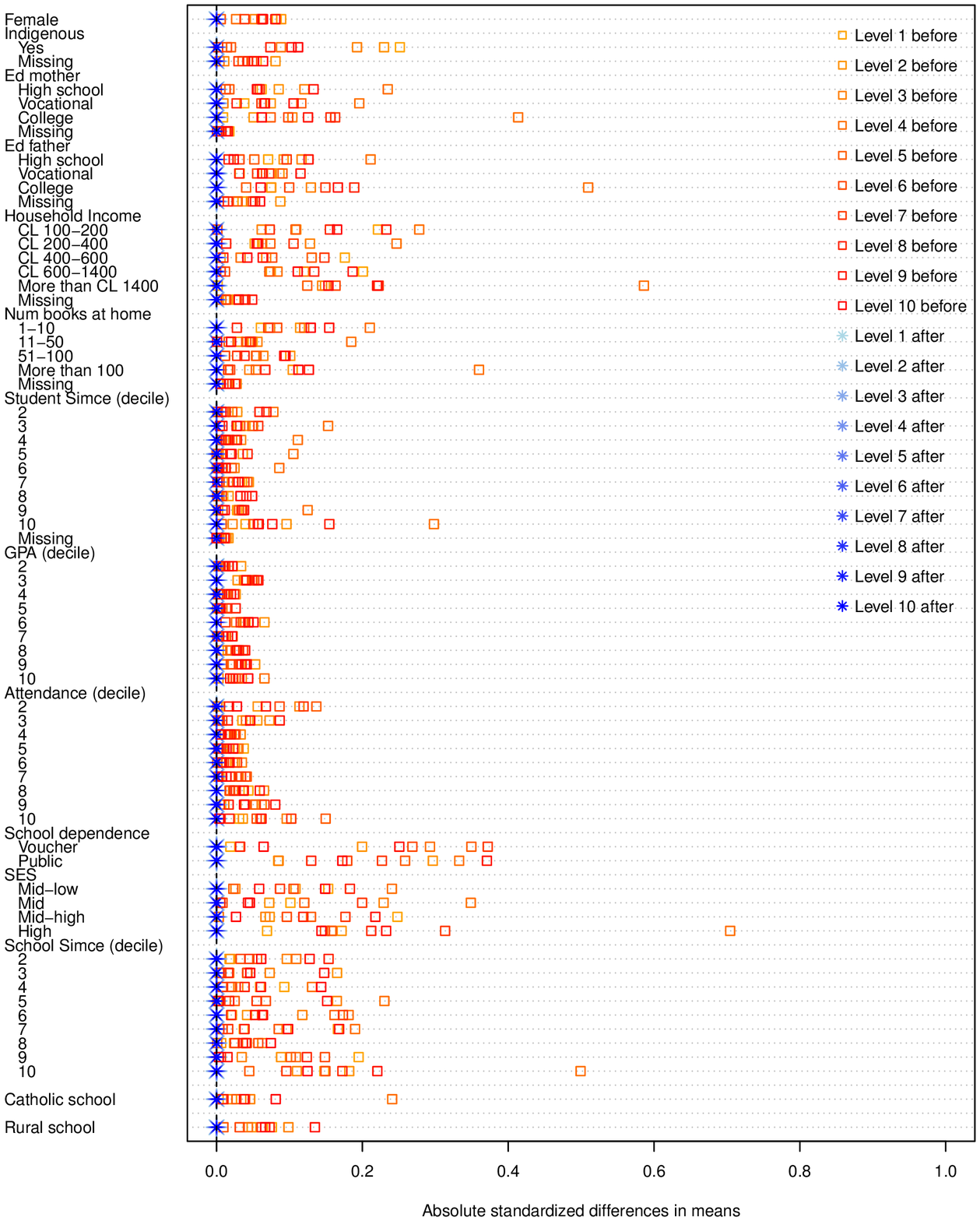}}
\label{fig:balance}
\end{figure}

\subsection*{Appendix E: Effect estimates}

\begin{table}[H]
\begin{center}
\caption{Effect estimates and 95\% confidence intervals for different levels of exposure to the earthquake.  The point estimates are contrasts with respect to exposure level 1.  The 95\% confidence account for multiple comparisons.\label{outcomes_cc}}
\subfloat[3 exposure levels]{
\begin{adjustbox}{scale=0.67}
\begin{tabular}{ccc}
  Exposure level & Attendance (\%) & PSU score\\
  \hline
2 &-1.00 & 11.00 \\ 
&  [-2.01,0.00] & [1.90,18.50] \\ 
3&  -11.85 & 8.00 \\ 
 & [-12.35,-11.00] & [1.50,14.10] \\ 
   \hline
\end{tabular}

\end{adjustbox}
}\\
\subfloat[5 exposure levels]{
\begin{adjustbox}{scale=0.67}
\begin{tabular}{ccc}
  Exposure level & Attendance (\%) & PSU score \\
  \hline
2&-0.00 & 6.00 \\ 
&  [-1.00,1.00] & [-0.60,14.00] \\ 
3&  -2.55 & 6.50 \\ 
 & [-4.00,-1.99] & [-0.10,13.50] \\ 
 4& -4.00 & 4.50 \\ 
  & [-5.01,-3.00] & [-3.00,11.00] \\ 
 5&  -12.95 & 8.50 \\ 
 & [-13.70,-12.09] & [0.90,15.10] \\ 
   \hline
\end{tabular}

\end{adjustbox}
}\\
\subfloat[10 exposure levels]{
\begin{adjustbox}{scale=0.67}
\begin{tabular}{ccc}
  Exposure level & Attendance (\%) & PSU score \\
  \hline
2 & 0.00 & 4.00 \\ 
&  [-0.01,1.01] & [-3.50,11.60] \\ 
3&  0.00 & 9.00 \\ 
&  [-1.00,1.01] & [0.90,17.10] \\ 
 4 & -1.00 & 10.00 \\ 
 & [-2.00,0.01] & [2.00,17.60] \\ 
 5 & -2.00 & 9.00 \\ 
  & [-2.16,-0.99] & [0.90,16.50] \\ 
 6 & -3.00 & 7.50 \\ 
 & [-4.01,-2.00] & [-0.50,14.50] \\ 
 7& -2.00 & -3.50 \\ 
 & [-3.01,-1.69] & [-3.60,11.00] \\ 
 8&  -6.10 & -4.00 \\ 
  & [-7.66,-5.00] & [-4.10,11.10] \\ 
 9& -11.80 & 9.00 \\ 
 &  [-12.85,-10.80] & [1.50,17.50] \\ 
 10&  -13.75 & 10.50 \\ 
  & [-14.56,-12.94] & [3.0,18.60] \\ 
  \hline
\end{tabular}

\end{adjustbox}
}
\end{center}
\end{table}


\end{document}